\documentclass{llncs}
\usepackage{tikz}
\usetikzlibrary{calc}
\usepackage{amsmath,amssymb}
\usepackage{framed}
\usepackage{fancyhdr,lastpage,verbatim,rotating,amsfonts,color}
\usepackage[normalem]{ulem}
\usepackage{bold-extra}
\usepackage{tikz,pgfplots}
\usepackage{graphicx}
\usepackage{graphics}
\usetikzlibrary{positioning}
\usepackage{algorithm}
\usepackage{algpseudocode}
\usepackage[mathscr]{eucal}
\usepackage{filecontents}
\usepackage{framed}
\usepackage{caption}
\usepackage{subcaption}
\usepackage{url,changepage}

\usepackage[colorlinks,linkcolor=blue,citecolor=blue,urlcolor=blue]{hyperref}

\newcommand{\boxplot}[7]{%
	\filldraw[fill=white,line width=0.2mm,draw=#7] let \n{boxxl}={#1-0.15}, \n{boxxr}={#1+0.15} in (axis cs:\n{boxxl},#3) rectangle (axis cs:\n{boxxr},#4);
	\draw[line width=0.2mm, color=red, very thick] let \n{boxxl}={#1-0.15}, \n{boxxr}={#1+0.15} in (axis cs:\n{boxxl},#2) -- (axis cs:\n{boxxr},#2);
	\draw[line width=0.2mm,color=#7] (axis cs:#1,#4) -- (axis cs:#1,#6);                                                                           
	\draw[line width=0.2mm,color=#7] let \n{whiskerl}={#1-0.025}, \n{whiskerr}={#1+0.025} in (axis cs:\n{whiskerl},#6) -- (axis cs:\n{whiskerr},#6);
	\draw[line width=0.2mm,color=#7] (axis cs:#1,#3) -- (axis cs:#1,#5);                                                                           
	\draw[line width=0.2mm,color=#7] let \n{whiskerl}={#1-0.025}, \n{whiskerr}={#1+0.025} in (axis cs:\n{whiskerl},#5) -- (axis cs:\n{whiskerr},#5);
}
\newcommand{\outlier}[3]{
	\draw[line width=0.2mm,color=#3] (axis cs:#1,#2) circle (1.5pt);
}

\usepackage{cleveref}
\crefname{chapter}{Chapter}{Chapters}
\crefname{section}{Section}{Sections}
\crefname{appendix}{Appendix}{Appendices}
\crefname{table}{Table}{Tables}
\crefname{figure}{Figure}{Figures}
\crefname{algorithm}{Algorithm}{Algorithms}
\crefname{theorem}{Theorem}{Theorems}
\crefname{lemma}{Lemma}{Lemmas}
\crefname{corollary}{Corollary}{Corollaries}
\crefname{remark}{Remark}{Remarks}

\usepackage{thmtools}
\usepackage{thm-restate}

\def\clap#1{\hbox to 0pt{\hss#1\hss}}
 \def\mathrlap{\mathpalette\mathrlapinternal} 

\def\mathrlapinternal#1#2{\rlap{$\mathsurround=0pt#1{#2}$}}

\newcommand{\OPT}{\operatorname{OPT}}

\newcommand{\LOP}{\textup{LMP}}
\newcommand{\flex}{\textup{flex}}
\newcommand{\fixed}{\textup{fix}}
\newcommand{\adap}{\textup{A}}
\newcommand{\nAdap}{\textup{N}}
\newcommand{\LOPfa}{\LOP^{\flex}_{\adap}}
\newcommand{\LOPfn}{\LOP^{\flex}_{\nAdap}}
\newcommand{\LOPxa}{\LOP^{\fixed}_{\adap}}
\newcommand{\LOPxn}{\LOP^{\fixed}_{\nAdap}}
\newcommand{\LPflex}{\textup{[LP]}}

\newcommand{\OPTfa}{\opt^{\flex}_{\adap}}
\newcommand{\OPTfn}{\opt^{\flex}_{\nAdap}}
\newcommand{\OPTxa}{\opt^{\fixed}_{\adap}}
\newcommand{\OPTxn}{\opt^{\fixed}_{\nAdap}}

\newcommand{\fmpN}{$\textup{FMP}_{\nAdap}$}
\newcommand{\fmpA}{$\textup{FMP}_{\adap}$}

\newcommand{\costN}{f_{\nAdap}}
\newcommand{\costA}{f_{\adap}}
\newcommand{\opt}{\textup{OPT}}
\newcommand{\optN}{\opt_{\nAdap}}
\newcommand{\optA}{\opt_{\adap}}
\newcommand{\SP}[2]{\operatorname{SP}\!\!_{#1}(#2)}

\newcommand{\procSP}{\textsc{findPaths}}

\pgfplotsset{filter discard warning=false}

\algrenewcommand{\algorithmicforall}{\textbf{for each}}

\def\ForEach{\ForAll}

\newcommand{\paths}{\mathcal{P}}

\DeclareMathOperator*{\argmin}{arg\,min}

\newcommand{\rom}[1]{\uppercase\expandafter{\romannumeral #1\relax}}
\definecolor{darkbrown}{RGB}{102, 51, 0}

\xdefinecolor{darkgreen}{RGB}{11, 143, 7}

\makeatletter
\def\therule{\makebox[\algorithmicindent][l]{\hspace*{.5em}\vrule height .75\baselineskip depth .25\baselineskip}}%

\newtoks\therules
\therules={}
\def\appendto#1#2{\expandafter#1\expandafter{\the#1#2}}
\def\gobblefirst#1{
  #1\expandafter\expandafter\expandafter{\expandafter\@gobble\the#1}}%
\def\LState{\State\unskip\the\therules}
\def\pushindent{\appendto\therules\therule}%
\def\popindent{\gobblefirst\therules}%
\def\printindent{\unskip\the\therules}%
\def\printandpush{\printindent\pushindent}%
\def\popandprint{\popindent\printindent}%

\algdef{SE}[FOR]{AlgForEach}{AlgEndFor}[1]
  {\printandpush\algorithmicforall\ #1\ \algorithmicdo}
  {\popandprint\algorithmicend\ \algorithmicforall}%
\algdef{SE}[IF]{AlgIf}{AlgEndIf}[1]
  {\printandpush\algorithmicif\ #1\ \algorithmicdo}
  {\popandprint\algorithmicend\ \algorithmicif}%
\algdef{C}[IF]{IF}{AlgElsIf}[1]
  {\popandprint\pushindent\algorithmicelse\ \algorithmicif\ #1\ \algorithmicthen}%
\algdef{SE}[PROCEDURE]{AlgProcedure}{AlgEndProcedure}[2]
   {\printandpush\algorithmicprocedure\ \textproc{#1}\ifthenelse{\equal{#2}{}}{}{(#2)}}%
   {\popandprint\algorithmicend\ \algorithmicprocedure}%
\algdef{SE}[WHILE]{AlgWhile}{AlgEndWhile}[1]
  {\printandpush\algorithmicwhile\ #1\ \algorithmicdo}
  {\popandprint\algorithmicend\ \algorithmicwhile}%
\makeatother

\title{Fare Evasion in Transit Networks}
\titlerunning{Fare Evasion in Transit Networks}

\newcommand{\tuaddress}{%
Institut f{\"u}r Mathematik, Technische Universit\"{a}t Berlin
}
  
\newcommand{\maastricht}{%
School of Business and Economics, Maastricht University
}

\newcommand{\chile}{%
Departamento de Ingenier\'ia Industrial, Universidad de Chile
}

\institute{\chile \and \maastricht \and \tuaddress}

\author{Jos\'e R.~Correa\inst{1} \and Tobias Harks\inst{2} \and Vincent\,J.C.\,Kreuzen\inst{2} \and Jannik
  Matuschke\inst{3}
}

\begin{document}

\maketitle
\pagenumbering{arabic}
\begin{abstract}
Public transit systems in urban areas usually require large state subsidies, primarily due to high fare evasion rates. In this paper, we study new models for optimizing fare inspection strategies in transit networks based on bilevel programming. In the first level, the leader (the network operator) determines probabilities for inspecting passengers at different locations, while in the second level, the followers (the fare-evading passengers) respond by optimizing their routes given the inspection probabilities and travel times.
To model the followers' behavior we study both a non-adaptive variant, in which passengers select a path a priori and continue along it throughout their journey, and an adaptive variant, in which they gain information along the way and use it to update their route.  For these problems, which are interesting in their own right, we design exact and approximation algorithms and we prove a tight bound of 3/4 on the ratio of the optimal cost between adaptive and non-adaptive strategies.
For the leader's optimization problem, we study a fixed-fare and a flexible-fare variant, where ticket prices may or may not be set at the operator's will. For the latter variant, we design an LP based approximation algorithm. Finally, using a local search procedure that shifts inspection probabilities within an initially determined support set, we perform an extensive computational study for all variants of the problem on instances of the Dutch railway and the Amsterdam subway network. This study reveals that our solutions are within $95\%$ of theoretical upper bounds drawn from the LP relaxation.
\end{abstract}

\section{Introduction\label{section:introduction}}

Fare evasion in transit systems causes significant losses to society. Recent studies revealed an annual loss of 70 million pounds for London's transit system~\cite{London10} and fare evasion rates of nearly $20\%$ in Santiago's transit system, which is subsidised with over 700 million dollars each year~\cite{transantiago}. As the installation of physical ticket barriers is not always possible, cost-efficient, or desirable, many transit systems, such as the Dutch and German railway and subway networks and the London bus system, rely on the honesty of customers and proper controlling of valid tickets \emph{on board}. Travelers who are caught without a ticket have to pay a fine, which is significantly larger than the ticket price.

In this paper, we study the optimization of fare inspection strategies in transit systems taking into account realistic models for the passenger's reaction. Our models are based on a \emph{bilevel optimization problem} (or Stackelberg game) on the network. In the first level, the leader, who strives to maximize the revenue from ticket sales and collected fines, determines for each edge the probability of controlling passing passengers, representing the frequency of inspections on that edge. In light of the limited number of inspectors available, we assume a global budget constraint on the sum of these probabilities. Given the inspection probabilities on edges, we assume passengers to act strategically and decide on whether or not to buy a ticket and along which path to travel based solely on their perceived cost. If a passenger chooses not to buy a ticket, her cost is thus the sum of the path length (expressed in monetary units) and the expected fine to be payed (which depends on the inspection probabilities along the path); on the other hand, if she buys a ticket, her cost is defined as the cost of the cheapest path considering both ticket price and distance measured in monetary units.

We perform a complete study differentiating between two different possibilities for the passengers' behavioral assumption as well as two different settings for the leader's decision problem. From the followers side, we consider both the case in which they are \emph{adaptive} and a \emph{non-adaptive}. In the latter, each passenger chooses a path at the beginning of her journey and continues along it, independent of whether or not she encounters an inspector. In the adaptive version, passengers adapt their behavior and consequently, when caught without a ticket, they continue on their shortest path only considering distance (since the fine includes a ticket to finish the trip). For the leader, we also study two different settings. In the \emph{fixed-fare} setting, we assume that ticket prices are fixed a priori, e.g., by governmental regulations, and the leader only sets the inspection probabilities subject to a budget constraint. In the \emph{flexible-fare} setting, we assume that the leader can additionally determine ticket prices.

\subsubsection{Our Contribution and Structure of the Paper.}

The main goal of this work is to provide a comprehensive study of fare evasion and fare inspection problems in transit networks. As mentioned above, we consider four versions of the problem by variating on whether the followers are adaptive or not, and wether the ticket prices are fixed or flexible. From a methodological viewpoint we tackle the problem from different angles, designing polynomial time algorithms and approximation schemes for the followers' problems, studying approximation algorithms and LP relaxations for the leader's problem, and using a local search heuristic to obtain high quality solutions for real-world instances.

After establishing the precise model in \cref{section:preliminaries}, we study the two variants of the followers' minimization problem in \cref{sec:followers}. These are natural extensions of the classic shortest path problem and exhibit interesting properties on their own which go beyond the particular application in this work. For the non-adaptive variant, we design a fully polynomial time approximation scheme
for general network topologies, exploiting similarities between this problem and the restricted shortest path problem (RSP). Interestingly we also obtain an exact polynomial time algorithm for series-parallel graphs, in sharp contrast to RSP, which is still $N\!P$-hard in these graphs. For adaptive followers, we obtain an exact polynomial time algorithm for general graphs using an optimal substructure property of the problem. Surprisingly, we further show that for arbitrary probability distributions, the optimal solution to the non-adaptive followers' minimization problem is at most a factor $4/3$ away from the optimal solution to the adaptive variant, and we also show that this bound is tight.

In \cref{sec:leader}, we turn our attention to the leader's maximization problem and prove that all four variants of this problem are strongly NP-hard. We then present an LP relaxation that yields a valid upper bound on the achievable profit for all four variants, and also obtain a $(1-1/e)$-approximation algorithm for the variant involving flexible ticket prices and non-adaptive followers. Combining this with the worst-case gap from the preceding section yields a  $\frac{3}{4}(1-1/e)$-approximation for the variant with adaptive followers. Finally, we develop a local search procedure that shifts inspection probabilities within an initially determined support set of edges. As candidate support sets, we use solutions from an LP relaxation, a minimum multicut, and a related mixed integer program~\cite{borndorfer}. 

In \cref{section:heuristics}, we demonstrate the applicability of our local search approach for
all four problem variants by conducting an extensive computational study comprising a total of $5600$ instances based on the networks of the Dutch railway system, the Amsterdam metro system, and randomly generated graphs. Our study reveals that the objective values of the computed solutions are within $95\%$ of the calculated upper bound on average.

\subsubsection{Related Work.} 
Bilevel network pricing problems have been considered extensively in the literature; see~\cite{BrotcorneLMS01,BrotcorneLMS08,Marcotte98}. In these works, 
the leader prices edges of some subgraph and the followers choose shortest paths with respect to edge costs defined as the sum of travel costs and prices.
This problem has been proved to be NP-hard; see~\cite{Marcotte98} and \cite{Bouhtou07,BriestHK12,RochSM05} for further hardness and approximability results. 

Bornd\"orfer et al.~\cite{borndorfer,borndorferAPRIL} investigated a problem similar to our variant of nonadaptive followers and fixed ticket prices. In contrast to our model, they assume the cost of the followers to be given by the \emph{sum} of travel costs and probabilities, reducing the followers minimization problem to a standard shortest path problem. Based on this simplification, the authors derive a compact mixed integer program (MIP) for the leader's optimization problem and solve it on several test instances derived from the German motorway network. In addition, they also study the case where the leader
wants to minimize the number of fare evaders. In contrast, we model the followers' response more precisely by considering the exact expected fine along a path, leading to a nonlinear aggregation of probabilities. In our computational study, we systematically compare the difference between both modeling approaches. We show that the quality of plain MIP solutions (as in~\cite{borndorfer,borndorferAPRIL}) is on average over all instances $5\%$ lower than our solutions when taking the exact followers response into account. For instances derived from the Amsterdam metro network, our solution improves even by $7.4\%$ on average, with an improvement of more than $20\%$ on some instances.

There is a body of literature on so-called network security games, where a defender can check a limited number of edges in a graph and the attacker computes from a given set of source and target nodes a path connecting a source and a target; see~\cite{conitzer,tambe,washburn}. The payoff for the attacker (and defender) has a binary characteristic since it only depends on whether or not the attacker has been caught. In Lin et al.~\cite{lin}, a single defender moves through the network \emph{over time} to detect attackers. The authors propose a linear program which solves the problem to optimality for small instances, and heuristics which solve the problem for larger instances within 1\% of optimality.

Yin et al.~\cite{yin} study a fare inspection model similar to ours. In contrast to our work, passengers are assumed to follow a fixed route through the network. In~\cite{BorndorferSS12}, the problem of computing tours of inspectors for a given distribution of traffic flows is considered. The authors derive an integer programming formulation and conduct a computational study on the German motorway network.

\section{The Model\label{section:preliminaries}}
We we are given a directed graph $G = (V, E)$
with costs $c : E \rightarrow \mathbb{Z}_{+}$ modeling the transit times on the edges (as monetary cost incurred to the passengers traveling along them). In order to prevent fare evasion, the network operator (the \emph{leader}) sets inspection probabilities $p_e \in [0,1]$ on the edges subject to the budget constraint $\sum_{e\in E} p_e\leq B$, where $B \geq 0$ corresponds to a limited number of ticket inspectors. We first describe the passengers' reaction to the chosen inspection strategy and then discuss the resulting revenue for the operator.

\subsubsection{The followers.}
Passengers (the \emph{followers}) are modeled by a set of \emph{commodities} $K$. For every commodity $i \in K$, a demand $d_i \geq 0$ is given that specifies the number of passengers with origin $s_i \in V$ and destination $t_i \in V$. 
Passengers of commodity $i$ can either buy a ticket at price $T_i \geq 0$ or choose a path without paying the ticket. If a passenger is caught without a valid ticket he has to pay a fine $F \geq 0$ satisfying $F \geq T_i$. As in most public transport systems (e.g., the Dutch or German sub- and railway networks) the fine includes the ticket price and enables the passenger to continue his trip.
Passengers are assumed to be rational, deciding purely based on their personal costs, which is expressed as the sum of travel times along the chosen path and the monetary cost. In the following, we differentiate between a \emph{non-adaptive} and an \emph{adaptive} variant of the followers response. In both settings, passengers of commodity $i$ can either chose to pay the ticket and follow the shortest path w.r.t.~$c$, or decide to evade the fare, choosing a path $P \in \mathcal{P}_i$ that minimizes a variant-dependent cost function (where $\mathcal{P}_i$ denotes the set of $s_i$-$t_i$-paths).

\paragraph{Non-adaptive followers.} 
In the non-adaptive variant, passengers are assumed to choose a route before the start of their trip and continue along it independent of whether or not they encounter an inspector.
For a path $P$, we denote its total travel time by $c(P) := \sum_{e \in P} c_e$, and we denote the probability that no inspector is encountered along $P$ by $\pi(P) := \prod_{e \in P} (1 - p_e)$. Given probabilities~$p$, the expected cost of path $P$ for passenger $i$ in the non-adaptive setting is
\[f_{\nAdap,p,i} (P) := c(P) + \left(1 - \pi(P) \right) \cdot F.\]
This variant is plausible under the assumption that passengers have determined their route beforehand (via checking the route and timetables). We denote the corresponding optimization problem by \fmpN.

\paragraph{Adaptive followers.}
In the adaptive variant, once  a passenger is caught, he will continue his trip
along the shortest path (w.r.t.~$c$). Letting $\SP{c}{v, w}$ denote the shortest path distance from $v$ to $w$ w.r.t.~$c$, the follower thus tries to find a path $P = (e_1, \dots, e_k) \in \mathcal{P}_i$ with $e_i = (v_i, v_{i+1})$ minimizing the expected cost
\[f_{\adap, p, i}(P) := \sum_{i = 1}^{k} \prod_{j = 1}^{i-1} \big(1 - p_{e_j}\big) \cdot \Big(c_{e_i} +  p_{e_i} \big(F + \SP{c}{v_{i+1}, t}\big)\Big).\]
Note that in this formula, the $i$th summand corresponds to the event of $v_i$ being reached without inspection, in which case $e_i$ is traversed next.
This variant assumes that passengers know the shortest paths from all stations to their destination in a network, e.g., by calculating it using a portable computing device. We denote the corresponding optimization problem by \fmpA.

\subsubsection{The leader.}
The leader's problem can be defined as a bilevel problem, where the leader sets probabilities on the edges to which the followers respond by solving their individual optimization problems. While we will always assume the fine $F$ to be fixed in the problem input (fines are commonly determined by the legislation~\cite{nloverheid}), we will consider both the scenario where ticket prices are \emph{fixed} a priori and the scenario where they are \emph{flexible} and can be determined by the leader as part of his optimization problem. 
Combining this with the two different models for the followers' reaction, we obtain four different versions of the \emph{leader's maximization problem}, which we will denote by $\LOP^{L}_{X}$, with $L \in \{\fixed, \flex\}$ and $X \in \{\textup{N}, \textup{A}\}$ specifying the ticket pricing variant and the behaviour of the followers, respectively. Accordingly, we denote by $\Gamma^{L}_{X, i}(p)$ the leader's revenue 
per passenger
received from commodity $i$ when choosing inspection probabilities $p$. The leader thus wants to maximize $\sum_{i \in K} d_i \Gamma^{L}_{X, i}(p)$.

\paragraph{Fixed Fares.}
In the fixed-fare setting, the revenue received from passenger $i$ is either the ticket price or the expected revenue from collecting fines, i.e.,
\[\Gamma^{\fixed}_{X, i}(p) := \max \big\{F\cdot \left(1- \pi(P) \right) \; : \; P^* \in \argmin_{P\in\paths_i \cup \{P_i^T\}}f_{X, p, i}(P) \big\},\]
where $P_i^T$ denotes a special path representing the option of paying the ticket with $f_{X, p, i}(P^T_i) := \SP{c}{s_i, t_i} + T_i$ and $\pi(P^T_i) := 1 - T_i/F$.

\paragraph{Flexible Fares.}
When the leader is allowed to determine the ticket prices, he can ensure that every passenger chooses to buy a ticket by setting the prices sufficiently low. The maximum value a passenger is willing to pay for a ticket is
\[\Gamma^{\flex}_{X, i}(p) := \min \{ f_{X,p,i}(P) - \SP{c}{s_i, t_i} \; : \; P \in \mathcal{P}_i\} ,\]
because for this ticket price, the cost of paying the ticket and traveling along the shortest path is at most the expected cost of his best fare evasion option.
            
\section{The Followers' Minimization Problem}\label{sec:followers}
In this section, we design efficient algorithms of both variants of the followers' minimization problem. Throughout this section, we assume we are given the graph $G = (V, E)$, the start $s \in V$ and destination $t \in V$ of a given follower, costs $c$, probabilities $p$, and the fine $F$.
		
	\subsubsection{Non-adaptive Followers' Minimization Problem.}
		
	The non-adaptive version of the followers' minimization problem is related to the \emph{restricted shortest path problem} (RSP): If $P^*$ is an $s$-$t$-path minimizing $\costN$, then $P^*$ also maximizes $\pi(P^*)$ among all paths with $c(P) \leq C^* := c(P^*)$. By discretizing the set of possible values for $C^*$, this relation can be used to derive an FPTAS for {\fmpN}; see \cref{app:fptas} for a detailed proof.
	
	\begin{restatable}{theorem}{restateThmFPTAS}\label{thm:fptas}
		There is an algorithm for {\fmpN} that computes a $(1 + \varepsilon)$-approximate solution in time polynomial in $\frac{1}{\varepsilon}$, $\log c_{\max}$, $|V|$, and $|E|$.
	\end{restatable}
	
	It is well-known that RSP is $N\!P$-hard even on a very restricted subclass of series-parallel graphs. Given the similarity of the two problems, it seems natural to expect this hardness result to carry over to {\fmpN}. Surprisingly, it turns out that {\fmpN} can be solved exactly in polynomial time on series-parallel graphs.

	\begin{theorem}\label{thm:sepa}
	  There is an algorithm that computes an optimal solution to instances of {\fmpN} with series-parallel graphs in polynomial time.
	\end{theorem}
	
\paragraph{Sketch of proof.} The algorithm makes use of the inductive definition of a series-parallel graph as the series or parallel composition of two smaller series-parallel graphs. It recursively decomposes the graph and computes for each of the two subgraphs a set of candidate subpaths through them. It can be shown that an optimal path through the entire graph can be obtained from a composition of those candidate subpaths. By carefully bounding the probability of an inspection on the entire path depending on the parameters of the corresponding subpaths, we can restrict ourselves to a polynomial number of combinations of subpaths. A complete description and analysis of the algorithm can be found in \cref{app:sepa-alg}.

\paragraph{An exact algorithm for general graphs.}
We also provide an exact, however non-polynomial, algorithm for {\fmpN} on general graphs. The algorithm is based on a bidirectional vertex labeling approach and enumerating all Pareto-optimal paths, where a path $P$ is dominated by a path $P'$ if both $c(P) \geq c(P')$ and $\pi(P) \leq \pi(P')$. The algorithm is described in detail in \cref{app:exact-fmpN}. In order to speed up the enumeration, we use an adaptation of \emph{contraction hierarchies}~\cite{GeisbergerTS}, a preprocessing technique that performs node contractions, enabling a more efficient traversal of the paths in the graph. This technique is known to work especially well on sparse graphs, a common property of transit networks.

	\subsubsection{Adaptive Followers' Minimization Problem.}	
	The adaptive version of the followers minimization problem can be solved in polynomial time using a label setting algorithm. The algorithm makes use of the following observation on the structure of the cost function $\costA$, with $P[v, w]$ denoting the $v$-$w$-subpath of $P$.
	
		\begin{restatable}{lemma}{restateLemAdaptiveCost}\label{obs:adaptive_cost}
	  $\costA(P) \ = \ \costA(P[s, v]) \ + \pi(P[s, v]) \cdot \costA(P[v, t])$.
  \end{restatable}
  
  As an immediate consequence of \cref{obs:adaptive_cost}, we can deduce that every suffix of an optimal path must also be an optimal path.\footnote{Note however that the same is not true for prefixes of the optimal path, as the arcs of $P[s, v]$ also appear in the second summand.}
  
  \begin{corollary}
    Let $P^*$ be an $s$-$t$-path minimizing $\costA$ and let $v \in V(P^*)$. Then $P^*[v, t]$ is a $v$-$t$-path minimizing $\costA$.
  \end{corollary}
  
  In the spirit of Dijkstra's shortest path algorithm~\cite{dijkstra1959note}, our algorithm iteratively computes the cost of an optimal $v$-$t$-path for some vertex $v$. It maintains the set $S$ of vertices for which an optimal path has been computed and a value $\phi(v)$ for every vertex $v \in V$, denoting the cost of the cheapest $v$-$t$-path found so far. In every iteration, a vertex $w \in V \setminus S$ with minimum value $\phi(w)$ is added to $S$ and the labels $\phi(v)$ of vertices $v$ with $(v, w) \in E$ are updated accordingly. The complete algorithm and its analysis can be found in \cref{app:exact-fmpA}.

	\begin{restatable}{theorem}{restateThmAdaptiveFollower}\label{thm:adaptive_follower}
	  There is an algorithm that solves {\fmpA} in time $O(|E| + |V| \log |V|)$.
	\end{restatable}
	
	\subsubsection{The Impact of Adaptivity.}
	
	We now prove a tight upper bound of $4/3$ on the ratio of
	the optimal cost between non-adaptive and adaptive strategies.
	
	\begin{theorem}\label{thm:price_of_nonadaptability}
	  Let $\optN$ be the cost of an optimal non-adaptive solution and $\optA$ be the cost of an optimal adaptive solution. Then $\optN \leq \frac{4}{3}\optA$.
	\end{theorem}
	
	\begin{proof}
			  Let $P$ be a path that minimizes $\costA$. Observe that
	  \[
		\optA = \costA(P) \geq \pi(P) c(P) + (1 - \pi(P))(\SP{c}{s,t} + F),\label{observation:fmpa}
		\]
	  where the inequality is due to the fact that in any case, the passenger has at least to traverse a shortest path from $s$ to $t$. Furthermore,
	  \[\optN \leq \min \left\{c(P) + (1 - \pi(P))F,\ \SP{c}{s,t} + F \right\}\]
	  as both $P$ and a shortest path from $s$ to $t$ are feasible non-adaptive solutions.
	  Thus,
	  \begin{align*}
	  	\frac{\optN}{\optA} \ \leq \ \max_{p \in [0, 1], \, F \geq 0 \atop 0 \leq S \leq C} \frac{\min \left\{C + pF,\ S + F \right\}}{(1 - p)C + p (S + F)} \ \leq \ \max_{p \in [0, 1]} \frac{1}{1 - p + p^2} \ \leq \ \frac{4}{3},
	  \end{align*}
	  where the second inequality comes from the fact that the maximum is attained for $C + pF = S + F$. \qed
	\end{proof}
	The following example shows that the bound given in \cref{thm:price_of_nonadaptability} is tight.
	
	\begin{example}
	  Let $V = \{s, v, t\}$ and $E = \{e_0, e_1, e_2\}$ where $e_0$ goes from $s$ to $v$ and $e_1$ and $e_2$ are two parallel arcs from $v$ to $t$. The travel costs are $c_{e_0} = c_{e_1} = 0$ and $c_{e_2} = 1$, the inspection probabilities are $p_{e_0} = 1/2$, $p_{e_1} = 1$, and $p_{e_2} = 0$, and $F = 2$. Observe that there are only two $s$-$t$-paths $P_1 = (e_0, e_1)$ and $P_2 = (e_0, e_2)$ and that $\costN(P_1) = \costN(P_2) = 2$ and $\costA(P_2) = \frac{3}{2}$, yielding a ratio of $\frac{4}{3}$ between optimal non-adaptive and adaptive strategies.
	\end{example}
	
\begin{remark}
		Note that in the proof of \cref{thm:price_of_nonadaptability}, we did not make use of the fact that the probability $\pi(P)$ is determined by individual probabilities on the arcs of the network. Therefore, the bound of $\frac{4}{3}$ given by the theorem is still true for arbitrary probabilities $\pi(P)$ for each path $P$. In particular, the result still holds if inspections at specific edges are not independent events.
	\end{remark}
	
\section{The Leader's Maximization Problem}\label{sec:leader}
In this section, we discuss algorithms and complexity results for the leader's maximization problem. On the theoretical side, we derive $N\!P$-hardness for a restricted special case of the problem and an LP relaxation, which yields upper bounds on the profit for all four model variants. For the flexible-fare setting, we also obtain a constant factor approximation. On the practical side, we propose a local search procedure that, combined with initial solutions from the LP relaxation computes close-to-optimal solutions for all four variants of the problem.

\subsubsection{Complexity of the Leader's Problem.}
The following $N\!P$-hardness result can be derived by a simple reduction from the directed multicut problem; see \cref{app:hardness}.

\begin{restatable}{theorem}{restateThmHardness}\label{thm:hardness}
$\LOP^L_X$ for $L \in \{\fixed, \flex\}$ and $X \in \{\adap, \nAdap\}$ is strongly $N\!P$-hard, even when restricted to instances with $|K| = 2$ and $c \equiv 0$.
\end{restatable}

\subsubsection{LP Relaxation.}
Let $\OPT^L_X$ denote the value of an optimal solution to the corresponding version of the leaders maximization problem for $L \in \{\flex, \fixed\}$ and $X \in \{\adap, \nAdap\}$. The following lemmas relate these values to one another.

\begin{restatable}{lemma}{restateLemOptFlexFix}\label{lem:opt_val1}
$\OPTfn \geq \OPTxn$ and $\OPTfa \geq \OPTxa$.
\end{restatable}

\begin{restatable}{lemma}{restateLemOptAdap}\label{lem:opt_val2}
$\frac{4}{3} \OPTfa \geq \OPTfn \geq \OPTfa$.
\end{restatable}

In order to obtain the LP relaxation, we will make use of a linearization approach, which is based on the following classic approximation.

\begin{restatable}{lemma}{restateLemSumApprox}\label{lem:sum_approximation}
$\textstyle 1-\pi(P)\ \leq \ \min\big\{\sum_{e \in P} p_e, \ 1\big\} \ \leq \ \frac{1}{1 - \operatorname{e}^{-1}} \left(1-\pi(P)\right)$.
\end{restatable}

Using \cref{lem:sum_approximation}, we replace the term $1 - \pi(P)$ in the followers' objective function by $\sum_{e \in P} p_e$. Note that after this replacement, the non-adaptive version of FMP corresponds to a classic shortest path problem. Using the dual of the shortest path LP, we derive the following LP relaxation for $\LOPfn$.
\begin{alignat*}{3}
	\LPflex \quad \max \quad 				&& \sum_{i \in K} d_i (y_i(t_i) \mathrlap{- y_i(s_i) - \SP{c}{s_i, t_i})} & \\
  \text{s.t.} \quad && \sum_{e \in E} p_e & \ \leq \ B & \\
										&& y_i(w) - y_i(v) & \ \leq \ c_e + F p_e & \qquad \forall\; i \in K, \; e = (v, w) \in E \\
										&& p_e & \in [0, 1] & \forall\; e \in E
\end{alignat*}

The value $\OPT_\textup{LP}$ of an optimal solution to $\LPflex$ yields an upper bound to all four variants of the leader's maximization problem.

\begin{restatable}{lemma}{restateLemLpUpper}\label{lem:lp_upper_bound} $\OPT_\textup{LP} \geq \OPT^L_X$ for all $L \in \{\flex, \fixed\}$ and $X \in \{\adap, \nAdap\}$.
\end{restatable}

Using \cref{lem:opt_val2,lem:sum_approximation}, we can also derive that using an optimal solution to $\LPflex$ yields approximation algorithms for $\LOPfn$ and $\LOPfa$. 

\begin{restatable}{theorem}{restateThmLpApprox}\label{thm:lopfn-approx}
	There is a $(1 - 1/\operatorname{e})$-approximation algorithm for $\LOPfn$.
\end{restatable}

\begin{restatable}{corollary}{restateCorLpApprox}\label{cor:lopfa-approx}
	There is a $\frac{3}{4}(1 - 1/e)$-approximation algorithm for $\LOPfa$.
\end{restatable}

	The analysis of the algorithm given in \cref{thm:lopfn-approx} is tight. A corresponding instance is described in \cref{app:LPrelax}, along with the proofs of the preceding results. There, it is also shown that $\LPflex$ does not yield any constant approximation guarantees for the fixed-fare setting.
	
\subsubsection{Local Search Framework.}
We conclude this section by presenting a general local search framework to compute close-to-optimal solutions to the leader's maximization problem. The approach can be applied to all four model variants by using the corresponding followers' response and leader's objective function.

\paragraph{The Algorithm.}
In addition to an instance of $\LOP_X^L$, the input for the algorithm consists of a candidate subset $S \subseteq E$ of the graph edges and an initial setting of probabilities on these edges. An improving move of the local search chooses two disjoint subsets $E^+,E^- \subset S$ with $|E^+|, |E^-| \leq k$ for an input parameter $k$. The probabilities on the edges in $E^-$ are uniformly decreased up to a total decrease of $\Delta' = \min \{\Delta, \sum_{e\in E^-} p_e,  \sum_{e\in E^+} 1 -  p_e\}$, where $\Delta > 0$ is an exponentially decreasing step size. Then, the probabilities on the edges in $E^+$ are uniformly increased up to a total increase of $\Delta'$. The framework then recomputes the followers' response, and accepts the move, if the leader's profit increases or reverts it, otherwise. The algorithm terminates, if the improvement in the objective is below a given threshold for a given number of consecutive iterations.

\paragraph{Initial Solutions.} The performance of the local search framework depends crucially on the choice of candidate edges and the corresponding initial solution. We tested several methods for generating such solutions. The first is solving the LP relaxation {\LPflex} and using the support of the resulting solutions as candidate edge set. The second is computing a minimum cardinality directed multicut in the graph separating all terminal pairs, and then distributing the budget uniformly among the arcs in the cut. In addition, we also used the solutions from the MIP formulation given in~\cite{borndorfer} for the fixed price setting.

\section{Computational Study}\label{section:heuristics}
In this section, we present an extensive computational study on a broad set of realistic instances, assessing the solution quality of our local search approach and the impact of our new modelling approach compared to existing models. See \cref{sec:appendix_c} for a complete overview of our computational results and details of implementation and generation of instances.

\subsubsection{Test Instances.}
Experiments were performed on two different instance sets. The first set comprises the complete networks of the Amsterdam subway system and the Dutch railway system, as well as a subnetwork of the latter restricted to major transport hubs. The data was acquired from Dutch Railways~\cite{nlsite}. Due to privacy regulations, no real-world passenger data was available. Therefore, for each network, we generated ten instances each with $25$, $50$, $100$, and $200$ commodities by choosing pairs of vertices uniformly at random, and drawing the corresponding demands uniformly at random from the interval $[1, 50]$. Ticket prices were determined using a linearized variant of the formula used in official regulations, and travel costs (in monetary units) were calculated from actual travel times with a conversion rate of $0.132$~Euro per minute travel time; see~\cite{nlsite}. The
second instance set comprises randomly generated planar graphs exhibiting characteristics similar to those of real-world networks~\cite{randomPTN}. The graphs were generated using an approach similar to the one in~\cite{denise}. We generated ten graphs for each possible combination of $|V|, |K| \in \{25, 50, 100, 200\}$. For all instances, we tested $20$ different values of budgets from the range of $0.2$ to $25$. 
In total this yields $800$ instances for each of three real-world networks and each of the four graph size classes of the randomized set, yielding $5600$ instances in total.

\begin{table}[t]
	\centering
\scalebox{1}{
\begin{tabular}{r@{\hspace{0.4cm}}c@{\hspace{0.2cm}}c@{\hspace{0.2cm}}c@{\hspace{0.2cm}}c@{\hspace{0.4cm}}c@{\hspace{0.2cm}}c@{\hspace{0.2cm}}c@{\hspace{0.4cm}}c@{\hspace{0.2cm}}c@{\hspace{0.4cm}}c@{\hspace{0.2cm}}c}
		\hline
		&\multicolumn{4}{c}{$\LOPxn$}&\multicolumn{3}{c}{$\LOPxa$}&\multicolumn{2}{c}{$\LOPfn$}&\multicolumn{2}{c}{$\LOPfa$}\\
		graph set &best&LP&LP$^{\text{LS}}$&$\Delta_{\text{MIP}}$&best&LP&LP$^{\text{LS}}$&LP&LP$^{\text{LS}}$&LP&LP$^{\text{LS}}$  \\
		\hline
		
		nlmajor &97.4&94.8&97.3&2.09&97.4&94.1&97.2&96.4&98.0&96.4&98.0 \\
		nlcomplete &94.6&91.6&93.2&1.28&93.8&90.9&92.6&97.8&97.9&97.7&97.9 \\
		adammetro &98.2&94.7&97.7&7.09&98.2&94.7&97.8&94.7&97.8&94.7&97.8 \\
		\hline
		small &97.0&92.1&96.3&4.54&96.8&88.1&96.0&95.3&97.5&95.3&97.4 \\
medium &96.3&90.7&95.2&4.79&96.0&88.9&95.0&95.5&97.5&95.4&97.3 \\
large &95.7&90.4&95.1&4.82&95.3&88.6&94.5&96.6&98.0&96.5&97.8 \\
huge &95.6&89.1&94.1&4.03&94.6&85.2&92.7&96.5&97.9&96.4&97.7 \\
			\hline
		\end{tabular}
	}
	\caption{Average ratios between solutions and upper bounds in percent. 'best' denotes the average over the best solutions found for each instance, LP denotes the solutions found by $\LPflex$, $\text{LP}^{\text{LS}}$ denotes the solution found by performing the local search heuristic on the LP solution, and $\Delta_{\text{MIP}}$ denotes the improvement of 'best' compared to the solution found by using the MIP from~\cite{borndorfer}.\label{table:gaps}}
	\end{table}

\subsubsection{Algorithms.}
As algorithms for computing start solutions of the local search, we tested the LP relaxation (\textbf{LP}) and a minimum cardinality multi-cut computed using a standard MIP formulation (\textbf{MC}). In order to assess the impact of the more precise followers' objective function in our model as compared to existing approaches, we additionally computed the mixed integer programming solutions from~\cite{borndorfer} (\textbf{MIP}).
The follower's response in the local search procedure was computed using the exact algorithms for the respective variants presented in \cref{sec:followers}. After initial experiments for fine-tuning the parameters of the framework, it turned out that restricting to $k = 1$, i.e., probability shifts from one edge to another, is already sufficient for obtaining close-to-optimal solutions within $30$ iterations. We also set the initial step length $\Delta$ to $0.1$, decreasing it by a factor of $0.9$ in every iteration.
	
	\begin{figure}[t]
	
\begin{minipage}{1.0\textwidth}
		\scalebox{0.77}{
		\begin{tikzpicture}
			\begin{axis}[title={$\LOPxn$}, width=5.8cm,
									 height=7cm, xmin=-0.500000, xmax=5.700000, 
									 ymin=0.8, ymax=1,
									 xtick={0,1,2,3,4,5}, xticklabels={Rnd (best),Rnd (MIP), ADAM (best),ADAM (MIP), NL (best),NL (MIP)},
									 x tick label style={rotate=25, anchor=north east}
									] 
				\boxplot{0.000000}{0.964852126}{0.940193958}{0.985853399}{0.891190442}{1}{black}
				\outlier{0.000000}{0.855251044}{black}
				\outlier{0.000000}{0.86755345}{black}
				\outlier{0.000000}{0.874679659}{black}
				\outlier{0.000000}{0.8789841}{black}
				\outlier{0.000000}{0.884345169}{black}
				\boxplot{1.000000}{0.915728423}{0.879378913}{0.962419674}{0.794851273}{1}{black}
				\boxplot{2.000000}{0.98970537}{0.972172818}{1}{0.923217581}{1}{black}
				\boxplot{3.000000}{0.918227376}{0.87164736}{0.945328055}{0.821484237}{1}{black}
				\boxplot{4.000000}{0.962455798}{0.944171533}{0.991230938}{0.887145548}{1}{black}
				\boxplot{5.000000}{0.946742162}{0.920469619}{0.983809858}{0.857631512}{1}{black}
				\end{axis} 
		\end{tikzpicture}
		
		\raisebox{0.3625cm}{
		\begin{tikzpicture}
			\begin{axis}[title={best for $\LOPxa$}, width=3.7cm,
									 height=7cm, xmin=-0.500000, xmax=2.700000, 
									 ymin=0.8, ymax=1,
									 xtick={0,1,2}, xticklabels={Rnd,ADAM,NL},
									 x tick label style={rotate=35, anchor=north east}
									] 
				\boxplot{0.000000}{0.959428982}{0.938865127}{0.983370154}{0.890316931}{1}{black}
				\outlier{0.000000}{0.86999942}{black}
				\outlier{0.000000}{0.871720892}{black}
				\outlier{0.000000}{0.884345169}{black}
				\boxplot{1.000000}{0.98970537}{0.972172818}{1}{0.923217581}{1}{black}
				\boxplot{2.000000}{0.956032959}{0.937498227}{0.992181653}{0.853326881}{1}{black}
				\end{axis} 
		\end{tikzpicture}
		
		\begin{tikzpicture}
			\begin{axis}[title={LP$^{LS}$ for $\LOPfn$}, width=3.7cm,
									 height=7cm, xmin=-0.500000, xmax=2.700000, 
									 ymin=0.9, ymax=1,
									 xtick={0,1,2}, xticklabels={Rnd,ADAM,NL},
									 x tick label style={rotate=35, anchor=north east}
									] 
				\boxplot{0.000000}{0.981764124}{0.957917435}{0.994987707}{0.905975395}{1}{black}
				\boxplot{1.000000}{0.989590553}{0.958044547}{1}{0.907144236}{1}{black}
				\boxplot{2.000000}{0.983956137}{0.96382588}{0.995954492}{0.913839772}{1}{black}
				\end{axis} 
		\end{tikzpicture}
		
		\begin{tikzpicture}
			\begin{axis}[title={LP$^{LS}$ for $\LOPfa$}, width=3.7cm,
									 height=7cm, xmin=-0.500000, xmax=2.700000, 
									 ymin=0.9, ymax=1,
									 xtick={0,1,2}, xticklabels={Rnd,ADAM,NL},
									 x tick label style={rotate=35, anchor=north east}
									] 
				\boxplot{0.000000}{0.98035251}{0.9569471}{0.994726763}{0.905071931}{1}{black}
				\boxplot{1.000000}{0.989590553}{0.958044547}{1}{0.907144236}{1}{black}
				\boxplot{2.000000}{0.983789382}{0.963796763}{0.995954492}{0.912343503}{1}{black}
				\end{axis} 
		\end{tikzpicture}
		}
		} 
	\end{minipage}
	\caption{Ratio of profits to upper bounds. The diagrams show the distribution of the ratios for instances in the respective test sets. The mark inside each box denotes the median, boxes represent lower and upper quartiles and the whisker ends show the minimum and maximum, respectively, apart from possible outliers marked by a cycle. The diagrams were plotted following the suggestions in~\cite{frigge1989some}.\label{fig:boxplots}}
	\end{figure}
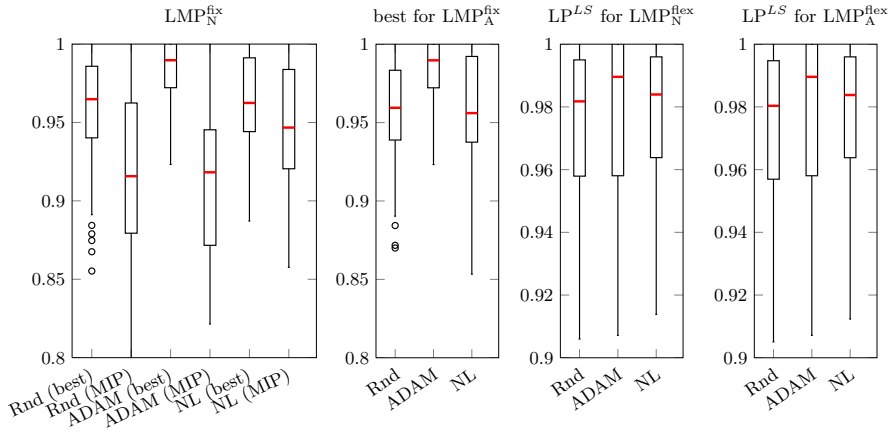

\subsubsection{Solution Quality.}
\cref{table:gaps} and \cref{fig:boxplots} show average gaps for all models on the test instances. With the exception of few instances with very high budget, results from $\text{LP}^{\text{LS}}$ consistently dominated those of $\text{MC}^{\text{LS}}$. We therefore omitted stating the results of the latter. Solutions for the fixed-fare variant are within $95\,\%$ of the upper bounds on average, while solutions for the flexible-fare variant are within $97.5\,\%$. This slight difference can be explained by the fact that the same upper bound was used for all four variants.

\subsubsection{Impact of the Budget Size.}
We also investigated the impact of budget sizes on the optimality gap and the achieved profit; see~\cref{figure:budgets} in \cref{sec:appendix_c} for a visualization.
In the range of $B\in [1,4]$, the ratio of obtained
profit and upper bound exhibits a bath tub curve behavior until it stabilizes 
at a ratio of $1$ for larger budgets (which is expected as for a large enough budget, all passengers can be forced to buy a ticket).
The achievable profit as a function of the budget is concave in all investigated examples,
which might be a universal property of the optimal
value function (as a function of the budget).

\subsubsection{Comparison with Existing Models.}
For the fixed-fare variant with non-adaptive followers, we can compare our modelling approach to that of~\cite{borndorfer}. We solved the MIP formulation proposed in~\cite{borndorfer} and computed the leader's profit resulting from the realistic response of followers (i.e., without linearization of their objective function). Comparing these solutions to the ones derived from our local search procedure, we observed that our approach yields an increase in profit of about $5\%$ on average on the randomly generated instances, about $7.5\%$ on the metro, and about $2\%$ on the railway network. In fact for some of the metro instances, the increase exceeds $20\%$; see \cref{table:gaps,fig:boxplots}.

\subsubsection{Acknowledgements.} We thank Ralf Bornd\"orfer and Elmar Swarat for giving an inspiring talk on the topic of toll enforcement that motivated our study.

\bibliographystyle{plain}
\bibliography{arxiv}

\clearpage
\appendix

\renewcommand\thelemma{\thesection\arabic{lemma}}

\section{Appendix for the Followers' Minimization Problem (\cref{sec:followers})}
\setcounter{lemma}{0}
\subsection{An FPTAS for {\fmpN} (Proof of \cref{thm:fptas})}\label{app:fptas}

\restateThmFPTAS*
\begin{proof}
	  Define $c_{\max} := \max_{e \in E} c_e$, $k := \lceil \log_{1 + \varepsilon} |V|c_{\max} \rceil$, and $C_i := (1 + \varepsilon)^i$ for $i \in [k]$. For each $i \in [k]$, we use the restricted shortest path algorithm by Hassin~\cite{Hassin:1992} to compute a path $P_i$ with $c(P_i) \leq (1 + \varepsilon)C_i$ and $\pi(P_i) \geq \pi(P)$ for all $P$ with $c(P) \leq C_i$.
	  Choose a path $P'$ from $P_0, \dots, P_k$ minimizing $\costN(P')$ among all $s_i$-$t_i$-paths of the graph. Let $P^*$ be an $s$-$t$-path minimizing $\costN$ and define $i^* := \min \{i \in [k] \, : \, c(P^*) \leq C_{i}\}$. Then $c(P_{i^*}) \leq (1 + \varepsilon)^2 c(P^*)$ and $\pi(P_{i^*}) \geq \pi(P^*)$ and thus $\costN(P') \leq \costN(P_{i^*}) \leq (1+\varepsilon)^2 \costN(P^*)$. \qed
	\end{proof}

\subsection{An Exact Polynomial Time Algorithm for {\fmpN} on Series-Parallel Graphs (Proof of \cref{thm:sepa})}\label{app:sepa-alg}

Our algorithm makes use of the inductive definition of a series-parallel graph as the series or parallel composition of two smaller series-parallel graphs. It recursively decomposes the graph and computes for each of the two smaller graphs a set of paths together with corresponding intervals. These paths and intervals will fulfil the following property: If an optimal path $P^*$ traverses the subgraph $G'$, then $\pi(P^*)$ must be in one of the intervals of the subgraph and the subpath of $P^*$ in $G'$ must have the same cost and inspection probability as the path corresponding to the interval. We can thus replace this part of $P^*$ by the corresponding path without changing the objective function value. This implies that for the original graph $G$, one of the polynomially many paths computed by the procedure must be optimal.

We start with a simple observation, which is true for {\fmpN} in any graph and will be useful in our algorithm.

\begin{lemma}\label{lem:parallel_bounds}
		Let $P_1$ and $P_2$ be two $s$-$t$-paths and let \mbox{$v, w \in V(P_1) \cap V(P_2)$}. For $i \in \{1, 2\}$ define $c_i := c(P_i[v, w])$ and $\pi_i := \pi(P_i[v, w])$. Assume $c_1 > c_2$ and $\pi_1 > \pi_2$. 
		\begin{itemize} 
			\item If $P_1$ minimizes $\costN$ among all $s$-$t$-paths, then 
				$\pi(P_1) \geq \dfrac{c_1 - c_2}{\pi_1 - \pi_2} \cdot \dfrac{\pi_1}{F}$.
			\item If $P_2$ minimizes $\costN$ among all $s$-$t$-paths, then 
				$\dfrac{c_1 - c_2}{\pi_1 - \pi_2} \cdot \dfrac{\pi_2}{F} \geq \pi(P_2)$.
		\end{itemize}
	\end{lemma}
	
	\begin{proof}
		Consider the path $P' := P_1[s, v] \circ P_2[v, w] \circ P_1[w, t]$. As $P_1$ minimizes $\costN$ we have
	\[0 \leq \costN(P') - \costN(P_1) = c_2 - c_1 + (1 - \frac{\pi_2}{\pi_1}) \pi(P_1) F.\]
	This immediately implies the claimed lower bound on $\pi(P_1)$. The upper bound for $\pi(P_2)$ follows analogously.\qed
	\end{proof}
	
\begin{algorithm}
		\caption{Algorithm for {\fmpN} on Series-parallel Graphs\label{alg:fmpN}}
		$(\mathcal{P}, \ell, u) \leftarrow \procSP(G, s, t, c, p, F)$\\
		Choose $P' \in \operatorname{argmin}_{P \in \mathcal{P}} \costN(P)$.\\
		\Return $P'$\\
		\begin{algorithmic}
			\Procedure{\procSP}{$G = (V, E), s, t, c, p, F$}
				\If{$E = \{(s, t)\}$}
					\State $P \leftarrow (s, t)$; $\ell(P) = 0$; $u(P) = 1$
					\State $\mathcal{P} = \{P\}$
				\ElsIf{$G$ is series composition of $G_1$ and $G_2$ joined at $v$}
					\State $(\mathcal{P}_1, \ell_1, u_1) \leftarrow \procSP(G_1, s, v, c, p, F)$
					\State $(\mathcal{P}_2, \ell_2, u_2) \leftarrow \procSP(G_2, v, t, c, p, F)$
					\State $\mathcal{P} \leftarrow \emptyset$
					\ForEach{$P_1 \in \mathcal{P}_1, P_2 \in \mathcal{P}_2$ with $[\ell_1(P_1), u_1(P_1)] \cap [\ell_2(P_2), u_2(P_2)] \neq \emptyset$}
						\State $P \leftarrow P_1 \circ P_2$
						\State $\ell(P) \leftarrow \max \{\ell_1(P_1), \ell_2(P_2)\}$
						\State $u(P) \leftarrow \min \{u_1(P_1), u_2(P_2)\}$
						\State $\mathcal{P} \leftarrow \mathcal{P} \cup \{P\}$
					\EndFor
				\ElsIf{$G$ is parallel composition of $G_1$ and $G_2$}
					\State $(\mathcal{P}_1, \ell_1, u_1) \leftarrow \procSP(G_1, s, t, c, p, F)$
					\State $(\mathcal{P}_2, \ell_2, u_2) \leftarrow \procSP(G_2, s, t, c, p, F)$
					\State $\mathcal{P} \leftarrow \mathcal{P}_1 \cup \mathcal{P}_2$
					\ForEach{$P \in \mathcal{P}$}
						\State $\ell(P) \leftarrow 0$; $u(P) \leftarrow 1$
						\ForEach{$P' \in \mathcal{P} \setminus \{P\}$}
							\If{$c(P) \leq c(P')$}
								\If{$\pi(P) \geq \pi(P')$}
									\State $\mathcal{P} \leftarrow \mathcal{P} \setminus \{P'\}$
								\ElsIf{$\pi(P) < \pi(P')$}
									\State $u(P) \leftarrow \min \Big\{u(P), \, \frac{c(P) - c(P')}{\pi(P) - \pi(P')} \frac{\pi(P)}{F}\Big\}$
								\EndIf
							\ElsIf{$c(P) > c(P')$}
								\If{$\pi(P) \leq \pi(P')$}
									\State $\mathcal{P} \leftarrow \mathcal{P} \setminus \{P\}$
								\ElsIf{$\pi(P) > \pi(P')$}
									\State $\ell(P) \leftarrow \max \Big\{\ell(P), \, \frac{c(P) - c(P')}{\pi(P) - \pi(P')} \frac{\pi(P)}{F}\Big\}$
								\EndIf
							\EndIf
						\EndFor
					\EndFor
				\EndIf\\
				\hspace{0.57cm}\Return $(\mathcal{P}, \ell, u)$
			\EndProcedure
		\end{algorithmic}
	\end{algorithm}
	
	The formal listing of \cref{alg:fmpN} comprises the procedure {\procSP}, which computes the intervals and corresponding paths by performing the series-parallel decomposition. Corresponding to this decomposition, it distinguishes three cases, which are described and analyzed in the proof of the following lemma.
		
	\begin{lemma}\label{lem:series-parallel}
	  Let $G' = (V', E')$ be a series-parallel subgraph of $G$ with start vertex $s'$ and end vertex $t'$. The set of paths $\mathcal{P}$ and corresponding interval bounds $\ell$ and $u$ computed by \procSP$(G', s', t', c, p, F)$ fulfill the following properties. 
	  \begin{itemize}
	  	\item $|\mathcal{P}| \leq |E'|$ 
	  	\item $[\ell(P), u(P)] \cap [\ell(P'), u(P')] = \emptyset$ for $P \neq P'$
	  	\item If $P^*$ is an $s$-$t$-path in $G$ minimizing $\costN$ with $s', t' \in V(P^*)$ and $P^*[s', t'] \subseteq E'$, then there is an $s'$-$t'$-path $P \in \mathcal{P}$ with $\pi(P^*) \in [\ell(P), u(P)]$, $c(P^*[s', t']) = c(P)$, and $\pi(P^*[s', t']) = \pi(P)$.
	  \end{itemize}
	\end{lemma}
	
	\begin{proof}
		We prove the lemma by induction on the algorithm. 
		\begin{itemize}
			\item Case 1: $G'$ is only the single edge $(s', t')$. In this case, {\procSP} returns the corresponding single-edge path and the statement of the lemma is trivially fulfilled.
			\item Case 2: $G'$ is the series composition of two series-parallel graphs $G_1$ and $G_2$ joined at a vertex $v$. Let $(\mathcal{P}_1, \ell_1, u_1)$ and $(\mathcal{P}_2, \ell_2, u_2)$ be the sets of paths and interval bounds computed by {\procSP} for $G_1$ and $G_2$, respectively. Every path $P \in \mathcal{P}$ in the set computed for $G'$ corresponds to the concatenation of two paths $P_1 \in \mathcal{P}_1$ and $P_2 \in \mathcal{P}_2$ such that the corresponding intervals $[\ell_1(P_1), u_1(P_1)]$ and $[\ell_2(P_2), u_2(P_2)]$ intersect. Every intersection starts either with a value $\ell_1(P_1)$ or a value $\ell_2(P_2)$. As the intervals defined by $\ell_1$ and $u_1$ ($\ell_2$ and $u_2$, respectively) are disjoint, every value of $\ell_2(P_2)$ (of $\ell_1(P_1)$, respectively) occurs at most in one intersection. Therefore $|\mathcal{P}| \leq |\mathcal{P}_1| + |\mathcal{P}_2| \leq |E'|$. 
			
			If $P^*$ traverses $G'$, it also has to traverse both $G_1$ and $G_2$. 
			By induction hypothesis, there is an $s'$-$v$-path $P_1 \in \mathcal{P}_1$ such that $P^*[s', v] \in [\ell_1(P_1), u_1(P_1)]$, $c(P^*[s', v]) = c(P_1)$, and $\pi(P^*[s', v]) = \pi(P_1)$. Furthermore, there is a $v$-$t'$-path $P_2 \in \mathcal{P}_2$ such that $P^*[v, t'] \in [\ell_2(P_2), u_2(P_2)]$, $c(P^*[v, t']) = c(P_2)$, and $\pi(P^*[v, t']) = \pi(P_2)$. This implies $[\ell_1(P_1), u_1(P_1)] \cap [\ell_2(P_2), u_2(P_2)] \neq \emptyset$ and therefore $P := P_1 \circ P_2$ is in $\mathcal{P}$ and fulfills $$c(P^*[s', t']) = c(P^*[s', v]) + c(P^*[v, t']) = c(P)$$ and $$\pi(P^*[s', t']) = \pi(P^*[s', v]) \cdot \pi(P^*[v, t']) = \pi(P).$$ Therefore, the statement of the lemma is true in this case.
			\item Case 3: $G'$ is the parallel composition of two series-parallel graphs $G_1$ and $G_2$. Let $(\mathcal{P}_1, \ell_1, u_1)$ and $(\mathcal{P}_2, \ell_2, u_2)$ be the sets of paths and interval bounds computed by {\procSP} for $G_1$ and $G_2$, respectively. The set $\mathcal{}P$ computed for $G'$ consists of a subset of the paths from $\mathcal{P}_1$ and $\mathcal{P}_2$. Therefore $|\mathcal{P}| \leq |\mathcal{P}_1| + |\mathcal{P}_2| \leq |E'|$. Let $P, P' \in \mathcal{P}$. Without loss of generality, we can assume $c(P) > c(P')$ and $\pi(P) > \pi(P')$, as otherwise either $P$ or $P'$ were removed from $\mathcal{P}$. Therefore 
			\[\ell(P) \geq \frac{c(P) - c(P')}{\pi(P) - \pi(P')} \frac{\pi(P)}{F} > \frac{c(P) - c(P')}{\pi(P) - \pi(P')} \frac{\pi(P')}{F} \geq u(P')\]
	by construction of $u$ and $\ell$. The corresponding intervals of $P$ and $P'$ are thus disjoint.
			
			 If $P^*$ is an $s$-$t$-path minimizing $\costN$ that traverses $G'$, then it either traverses $G_1$ or it traverses $G_2$. There thus is a path $P \in \mathcal{P}_1 \cup \mathcal{P}_2$ with $c(P) = c(P^*)$ and $\pi(P) = \pi(P^*)$. Note that $P$ is only removed from $\mathcal{P}$ if there is a path $P'$ with $c(P') \leq c(P)$ and $\pi(P') \geq \pi(P)$. By optimality of $P^*$, we must have $c(P') = c(P)$ and $\pi(P') = \pi(P)$. Thus, we can assume $P \in \mathcal{P}$ without loss of generality. It remains to show $\pi(P^*) \in [l(P), u(P)]$.
			  If $\ell(P) > 0$, then there is an $s'$-$t'$-path $P'$ in $G'$ such that $c(P) > c(P')$ and $\pi(P) > \pi(P')$ and $\ell(P) = \frac{c(P) - c(P')}{\pi(P) - \pi(P')} \frac{\pi(P)}{F}$. Thus, \cref{lem:parallel_bounds} implies $\pi(P^*) = \pi(P) \geq \ell(P)$. Analogously, $\pi(P^*) \leq u(P)$. \qed
		\end{itemize}
	\end{proof}
	
	\cref{lem:series-parallel} directly implies the optimality of the path computed by \cref{alg:fmpN} as well as the polynomiality of the algorithm, concluding the proof of \cref{thm:sepa}.
	
	\subsection{An Exact Algorithm for {\fmpN} on General Graphs}\label{app:exact-fmpN}

	\begin{algorithm}[t]
		\begin{algorithmic}
			\caption{Algorithm for {\fmpN} on General Graphs\label{alg:fmpNP}}
			\State $\mathcal{P} \leftarrow \emptyset$
			\State $S \leftarrow \{(s,P_0)\}$ \Comment $P_0$ is an empty $s$-$s$-path

			\While{$S \neq \emptyset$}
				\State Choose $(v,P) \in S$ minimizing $c(P)+F\cdot\left(1-\pi(P)\right)$.
				\State $S \leftarrow S \setminus (v, P)$
				\ForEach{$e=(v,w) \in \delta^+(v)$}
					\State $P' \leftarrow P \circ (e)$
					\State $c(P') \leftarrow c(P)+c_{e}$
					\State $\pi(P') \leftarrow \pi(P)\cdot (1-p_e)$
					\If{$c(P') < c(P^*)$ \textbf{or} $\pi(P') < \pi(P^*)$ for all $(w,P^*)\in S$}
						\If{$w = t$}
							\State $\mathcal{P} \leftarrow \mathcal{P} \cup P'$
						\Else
							\State $S \leftarrow S \cup \{(w,P')\}$
						\EndIf
					\EndIf
				\EndFor
			\EndWhile\\		
			\Return $P \in \mathcal{P}$ minimizing $c(P)+F\cdot\left(1-\pi(P)\right)$
		\end{algorithmic}
	\end{algorithm}
	
	We give a complete description of the exact (non-polynomial) algorithm for computing optimal solutions to {\fmpN} on general graphs. Our algorithm is based on enumerating all Pareto-optimal paths. A path $P$ is \emph{dominated} by a path $P'$ if both $c(P) \geq c(P')$ and $\pi(P) \leq \pi(P')$. For a vertex $v \in V$, a set of $s$-$v$-paths $\mathcal{P}$ is a \emph{dominating set}, if for every $s$-$v$-path $P$ in $G$ there is a path $P' \in \mathcal{P}$ that dominates $P$. Note that if $P'$ dominates $P$ then $\costN(P') \leq \costN(P)$. Therefore, a dominating set of $s$-$t$-paths contains an $s$-$t$-path minimizing $\costN$. The algorithm uses bidirectional label setting to a compute a dominating set for every vertex $v \in V$. It then returns an $s$-$t$-path which minimizes $\costN$. 
	
	As an important pre-processing step for the algorithm, we used \emph{contraction hierarchies}~\cite{GeisbergerTS}. Vertices in the graph are \emph{contracted} in a certain order, removing them while preserving the shortest paths distances between remaining vertices by adding \emph{shortcuts}, i.e. an edge $e=(u,w)$ is added when contracting a vertex $v$, if $u$ and $w$ are adjacent to $v$ and the shortest path from $u$ to $w$ is $((u,v),(v, w))$. Newly created shortcuts are added to the original graph. The sparseness of the resulting graph and the performance of the preprocessing depend crucially on the criteria used to determine the contraction order. In our setting an order that prioritizes contractions causing a large number of shortcuts turned out to be particularly efficient.

\subsection{An Exact Algorithm for {\fmpA} (Proofs of \cref{obs:adaptive_cost,thm:adaptive_follower})}\label{app:exact-fmpA}
		
		We give a complete description of the exact polynomial time algorithm for {\fmpA}. We start with the proof of \cref{obs:adaptive_cost}.
		
\restateLemAdaptiveCost*
	 \begin{proof}
  	Let $P = (e_1, \dots, e_k)$ with $e_i = (v_i, v_{i + 1})$ and assume $v = v_{k'}$ for some $k'$ with $1 \leq k' \leq k+1$. Observe that
  	\begin{align*}
  		\costA(P) \ = \ & \sum_{i = 1}^{k} \prod_{j = 1}^{i-1} \big(1 - p_{e_j}\big) \cdot \Big(c_{e_i} +  p_{e_i} \big(F + \SP{c}{v_{i+1}, t}\big)\Big)\\
  		= \ & \sum_{i = 1}^{k' - 1} \prod_{j = 1}^{i-1} \big(1 - p_{e_j}\big) \cdot \Big(c_{e_i} +  p_{e_i} \big(F + \SP{c}{v_{i+1}, t}\big)\Big)\\
  		& + \prod_{j = 1}^{k' - 1} \big(1 - p_{e_j}\big) \sum_{i = k'}^{k} \prod_{j = k'}^{i-1} \big(1 - p_{e_j}\big) \cdot \Big(c_{e_i} +  p_{e_i} \big(F + \SP{c}{v_{i+1}, t}\big)\Big)\\
  		= \ & \costA(P[s, v]) \ + \prod_{e \in P[s, v]} \!\!\! (1 - p_e) \cdot \costA(P[v, t]).\hspace{4.1cm}\qed
  	\end{align*}
  \end{proof}
  
  \begin{algorithm}
	\begin{algorithmic}
		\caption{Algorithm for {\fmpA}}\label{alg:fmpA}
		\LState $\phi(t) \leftarrow 0$
		\LState $\phi(v) \leftarrow \infty$ for all $v \in V \setminus \{t\}$
		\LState $S \leftarrow \emptyset$
		\AlgWhile{$s \notin S$}
		  \LState Choose $w \in V \setminus S$ minimizing $\phi$.
		  \LState $S \leftarrow S \cup \{w\}$
		  \AlgForEach{$e = (v, w) \in \delta^{-}(w)$}
		    \LState $\phi' \leftarrow c_e + p_e(\SP{c}{w} + F) + (1 - p_e) \phi(w)$
		    \AlgIf{$\phi' < \phi(v)$}
		      \LState $\phi(v) \leftarrow \phi'$
		      \LState $\operatorname{next}(v) \leftarrow e$
		    \AlgEndIf
		  \AlgEndFor
		\AlgEndWhile
		\LState $w \leftarrow s$; $P \leftarrow \emptyset$
		\AlgWhile{$w \neq t$}
		  \LState $e \leftarrow \operatorname{next}(w)$; $w \leftarrow \operatorname{head}(e)$
		  \LState Add $e$ to $P$.
		\AlgEndWhile\\
		\Return $P$
		\end{algorithmic}
	\end{algorithm}
		
	In the spirit of Dijkstra's shortest path algorithm~\cite{dijkstra1959note}, the algorithm iteratively computes the cost of an optimal $v$-$t$-path for some vertex $v$. It maintains the set $S$ of vertices for which an optimal path has been computed and a value $\phi(v)$ for every vertex $v \in V$, denoting the cost of the cheapest $v$-$t$-path found so far. In every iteration, a vertex $w \in V \setminus S$ with minimum value $\phi(w)$ is added to $S$ and the labels $\phi(v)$ of vertices $v$ with $(v, w) \in E$ are updated if a $v$-$t$-path consisting of $(v, w)$ and an optimal $w$-$t$-path is cheaper than the current value of $\phi(v)$. A complete listing is given as \cref{alg:fmpA}.
		
	\begin{lemma}\label{lem:fmpA-algorihm}
	  When vertex $v \in V$ is added to the set $S$ in \cref{alg:fmpA}, then there is an optimal $v$-$t$-path $P_v$ with $\costA(P_v) = \phi(v)$ starting with the edge $\operatorname{next}(v)$.
	\end{lemma}
	
	\begin{proof}
	  By contradiction assume the lemma is not true. Without loss of generality let $v$ be the first vertex added to $S$ that does not fulfill the statement of the lemma and  consider the moment when $v$ is added to $S$. Note that $v \neq t$, as $t$ fulfills the requirements with $\phi(t) = 0$ and $P_t = \emptyset$. Let $\operatorname{next}(v) = e = (v, w)$ for some $w \in S$. As $v$ is the first vertex violating the lemma, there is a $w$-$t$-path $P_w$ with $\costA(P_w) = \phi(w)$. Thus, $P_v := (e) \circ P_w$ is a $v$-$t$-path with cost 
	  \[\costA(P_v) = c_e + p_e (\SP{c}{w,t} + F) + (1 - c_e) \phi(w) = \phi(v).\]
	  Let $P'$ be any $v$-$t$-path, let $e' = (v', w')$ be the last edge on $P'$ with $v' \in V \setminus S$ and $w' \in S$. Note that 
	  \[c(P') + F \geq c_{e'} + p_{e'}) (\SP{c}{w',t} + F) + (1 - c_{e'})\phi(w') \geq \phi(v')\] as $\phi(w')$ denotes the cost of an optimal $w'$-$t$-path by induction hypothesis and therefore $\phi(w') \leq \SP{c}{w',t} + F$. Further note that $\phi(v') \geq \phi(v)$ by choice of $v$. In particular, this implies $\SP{c}{v,t} + F \geq \phi(v)$ by choosing $P'$ as a shortest $v$-$t$-path with respect to $c$. Thus, for arbitrary $P'$ again,
	  \[\costA(P') \geq \big(1 - \pi(P'[v,v'])\big) \cdot \big(\SP{c}{v,t} + F\big) + \pi(P'[v,v']) \cdot \phi(v') \geq \phi(v)\]
	  which proves that $P_v$ is optimal. \qed
	\end{proof}
	
	\restateThmAdaptiveFollower*
	
	\begin{proof}
	  The optimality of the path computed by \cref{alg:fmpA} follows immediately from \cref{lem:fmpA-algorihm}.
	  Shortest path distances from every vertex to $t$ can be pre-computed using Dijkstra's algorithm in time $O(|E| + |V| \log |V|)$. The remainder of the algorithm can be implemented using a Fibonacci heap for computing the vertex minimizing $\phi(v)$, guaranteeing the claimed running time.\qed
	\end{proof}
	
\section{Appendix for the Leader's Maximization Problem (\cref{sec:leader})}
\setcounter{lemma}{0}
\subsection{Complexity (Proof of \cref{thm:hardness})}\label{app:hardness}

We reduce from the \emph{minimum directed multicut problem}. An instance of this problem is given by a directed graph $G=(V,E)$ and $k$ commodities specified by origin-destination pairs $(s_i,t_i)_{i = 1,\dots,k}$. A multicut $M \subseteq E$ is a subset of edges such that $M \cap P \not= \emptyset$ for each path $P \in \mathcal{P}_i$ and each commodity $i \in \{1,\dots,k\}$. The task is to find a multicut of minimum cardinality. The corresponding decision problem is to decide for a given graph $G$ and an integer $q$ whether $G$ has a multicut of cardinality at most $q$. The directed multicut problem is strongly NP-hard even for $k=2$; see Garg, Vazirani, Yannakakis: \emph{Multiway Cuts in Directed and Node Weighted Graphs} (ICALP~2014).

\restateThmHardness*

\begin{proof}
 Consider an instance $ I=(G, (s_i,t_i)_{i =1,\dots,k}, q)$ of the directed
  multicut problem with $G=(V, E)$, $k=2$, and $q \in \mathbb{Z}_+$.
  We will construct an instance~$\hat I=(G,c,K,F,T,B)$
  of the leader's maximization problem as follows. We introduce two commodities, one for each pair $(s_i, t_i)$ with $i \in K := \{1, 2\}$. We set $c_e=0$ for all $e\in E$,
  $B=q$ and $T_1=T_2=F=1$.
  
  We denote by $\opt^L_X$ the value of an optimal solution to $\LOP^L_X$. Note that since the travel costs are all zero, $f_{\nAdap, p, i} = f_{\adap, p, i}$ for all $i \in K$ and any setting of probabilities $p \in [0, 1]^E$. Furthermore
 $\Gamma^\fixed_{\nAdap,i}(p) = \min_{P \in \mathcal{P}_i} \{f_{N,p,i}(P),\ T_i\} = \Gamma^\flex_{\nAdap,i}(P)$ as $\SP{c}{s_i, t_i} = 0$ for all $i \in K$. Therefore $\opt^\fixed_\nAdap = \opt^\flex_\nAdap = \opt^\flex_\nAdap = \opt^\flex_{\adap}$.
 
  We claim the following equivalence proving
   the theorem.
  \[\opt^L_X \geq 2 \Leftrightarrow \text{there exists a feasible
    multicut of cardinality $q$}.\]

Suppose that there is a multicut $M$ with $|M|\leq q$.
We then define $p_e=1$ for all $e\in M$ and $p_e=0$ for all $e\in E \setminus M$.
Note that $\sum_{e \in E} p_e = |M| \leq q$, implying that $p$ is a feasible solution.
Furthermore, every passenger encounters an inspector with probability $1$ because his path has to cross the multicut.
 We thus obtain $\opt^L_X \geq 2$.

Conversely, assume $p \in [0, 1]^E$ is a solution with profit $2$. This implies that $f_{X,i,p}(P) = 1$ for $i \in K$ and every path $P \in \mathcal{P}_i$. This is only possible if for every $P \in \mathcal{P}_i$, there is an $e \in P$ with $p_e = 1$. Therefore the set $M = \{e \in E \, : \, p_e = 1\}$ is a multicut with cardinality $|M| \leq \sum_{e \in M} p_e \leq q$.\qed
\end{proof}

\subsection{The LP Relaxation (Proofs of \cref{lem:opt_val1,lem:opt_val2,lem:sum_approximation,lem:lp_upper_bound,thm:lopfn-approx,cor:lopfa-approx})}\label{app:LPrelax}

\restateLemOptFlexFix*

\begin{proof}
Let $X \in \{\adap, \nAdap\}$. Note that $(1 - \pi(P)) F \geq f_{X, p, i}(P) - \SP{c}{s_i, t_i}$ for all $p \in [0, 1]^E$, $i \in K$ and any path $P \in \paths_i$. Therefore $\OPT^{\flex}_{X} \geq \OPT^{\fixed}_{X}$.\qed
\end{proof}

For the proofs of \cref{lem:opt_val2,cor:lopfa-approx}, we will additionally use the following stronger version of \cref{thm:price_of_nonadaptability}.

	\begin{lemma}\label{lem:adaptivity_ratio_leader}
		Let $\optN$ be the cost of an optimal non-adaptive solution and $\optA$ be the cost of an optimal adaptive solution. Then $\optN - \SP{c}{s, t} \leq \frac{4}{3}(\optA - \SP{c}{s, t})$.
	\end{lemma}
	
	\begin{proof}
	  Let $P$ be a path that minimizes $\costA$. Observe that
	  \[
		\optA = \costA(P) \geq \pi(P) c(P) + (1 - \pi(P))(\SP{c}{s,t} + F)\label{observation:fmpa}
		\]
	  where the inequality is due to the fact that in any case, the passenger has at least to traverse a shortest path from $s$ to $t$. Furthermore,
	  \[\optN \leq \min \left\{c(P) + (1 - \pi(P))F,\ \SP{c}{s,t} + F \right\}\]
	  as both $P$ and a shortest path from $s$ to $t$ are feasible non-adaptive solutions.
	  Thus,
	  \begin{align*}
	  	\frac{\optN - \SP{c}{s, t}}{\optA - \SP{c}{s, t}} \ & \leq \ \max_{p \in [0, 1], \, F \geq 0 \atop 0 \leq S \leq C} \frac{\min \left\{C + pF,\ S + F \right\} - S}{(1 - p)C + p (S + F) - S} \\
	  	& = \ \max_{p \in [0, 1], \, F \geq 0 \atop 0 \leq S \leq C} \frac{\min \left\{C - S + pF,\ F \right\}}{(1 - p)(C - S) + p F}.
	  \end{align*}
    If $C - S + pF > F$, then increasing $S$ also increases the value of the right-hand side function. If $C - S + pF < F$, then decreasing $F$ increases the value of the right-hand side function. Thus, if the right-hand side is maximized, then $C - S + pF = F$. Therefore
    \begin{align*}
      \frac{\optN - \SP{c}{s,t}}{\optA - \SP{c}{s,t}} \ \leq \ & \max_{p \in [0, 1] \atop F \geq 0} \frac{F}{F (1 - p + p^2)} \ \leq \ \frac{4}{3},
    \end{align*}
    which concludes the proof.\qed
	\end{proof}

\restateLemOptAdap*

\begin{proof}
Let $p \in [0, 1]^E$ be any solution to the $\LOP$ instance. Note that \cref{lem:adaptivity_ratio_leader} implies $\frac{4}{3} \Gamma^{\flex}_{\adap,i}(p) \geq \Gamma^{\flex}_{\nAdap,i}(p) \geq \Gamma^{\flex}_{\adap,i}(p)$ for all $i \in K$, and therefore $\frac{4}{3}\OPTfa \geq \OPTfn \geq \OPTfa$.\qed
\end{proof}

\restateLemSumApprox*

\begin{proof}
	We prove $\pi(P) = \prod_{e \in P} (1 - p_e) \geq 1 - \sum_{e \in P} p_e$ by induction on $|P|$. This is trivial for $|P| = 1$. For $|P| > 1$ observe that for any $e' \in P$ by induction hypothesis
	\[\prod_{e \in P} (1 - p_e) \geq \left(1 - \sum_{\mathrlap{\ e \in P \setminus \{e'\}}}\; p_e \right)(1 - p_{e'}) = 1 - \sum_{e \in P} p_e + p_{e'} \sum_{\mathrlap{\ e \in P \setminus \{e'\}}}\; p_e \geq 1 - \sum_{e \in P} p_e. \]
	This immediately implies the first inequality stated in the lemma.
	
	The second inequality of the lemma is trivially true if $\sum_{e \in P} p_e = 0$. Thus we assume the sum to be strictly positive without loss of generality. Define $\sigma := \min \{\sum_{e \in P} p_e, \ 1\}$ and observe that 
			\[\prod_{e \in P} (1 - p_e) \leq \left(1 - \frac{\sigma}{|P|}\right)^{|P|} \leq \operatorname{e}^{-\sigma}.\]
		Therefore
		\[\frac{1 - \pi(P)}{\sigma} \geq \frac{1 - \operatorname{e}^{-\sigma}}{\sigma} \geq \min_{x \in (0, 1]} \frac{1 - \operatorname{e}^{-x}}{x}.\]
		The right-hand side is decreasing in $x$ and therefore minimized for $x = 1$.\qed
\end{proof}

\restateLemLpUpper*

\begin{proof}
	Let $p \in [0, 1]^E$ be an optimal solution to $\LOPfn$. For every $i \in K$, set $y_i(s_i) = 0$ and $y_i(v) = \SP{c + Fp}{s_i, v}$ for all $v \in V \setminus \{s_i\}$. 
	It is easy to check that $(p, y)$ is a feasible solution to $\LPflex$ and that $y_i(t_i) - y_i(s_i) \geq \SP{c + pF}{s_i, t_i}$ for every $i \in K$. Therefore
	\begin{align*}
	y_i(t_i) - y_i(s_i) - \SP{c}{s_i, t_i} \ \geq \ & \min \left\{\sum_{e \in P} \left(c_e + Fp_e\right) - \SP{c}{s_i, t_i} \; : \; P \in \mathcal{P}_i\right\} \\
	\ \geq \ & \min \left\{\sum_{e \in P} c_e + (1 - \pi(P)) F - \SP{c}{s_i, t_i} \; : \; P \in \mathcal{P}_i\right\}\\
	 \ = \ & \Gamma^{\flex}_{\nAdap,i}(p)
	\end{align*}
	where the second inequality follows from \cref{lem:sum_approximation}. This implies that the optimal value of the LP is at least $\OPTfn$. By \cref{lem:opt_val1,lem:opt_val2}, $\OPTfn$ is as least as large as the optimal solution value of any of the other versions.\qed
\end{proof}

\newpage
\restateThmLpApprox*

\begin{proof}
	Let $(p, y)$ be an optimal solution to $\LPflex$. Note that $y_i(t_i) - y_i(s_i) \leq \SP{c + Fp}{s_i, t_i}$ and define $\lambda_i = y_i(t_i) - y_i(s_i) - \SP{c}{s_i, t_i}$. Then 
	\begin{align*}
  	\lambda_i \ &\leq \ \min \left\{\sum_{e \in P} (c_e + F p_e) - \SP{c}{s_i, t_i}, \ T_i \right\}\\
  	& \leq \ \sum_{e \in P} c_e - \SP{c}{s_i, t_i} + \min \left\{\sum_{e \in P} p_e,\ 1 \right\} \cdot F\\
  	& \leq \  \frac{1}{1 - \operatorname{e}^{-1}} \left(\sum_{e \in P} c_e - \SP{c}{s_i, t_i} + \left(1 - \prod_{e \in P} (1 - p_e)\right) \cdot F\right)
	\end{align*} 
	for every $P \in \mathcal{P}_i$ by \cref{lem:sum_approximation}. Thus, setting the probabilities according to $p$ yields a solution to $\LOPfn$ with profit 
	\[\sum_{i \in K} d_i\Gamma^{\flex}_{\nAdap,i}(p) \geq (1 - 1/\operatorname{e}) \sum_{i \in K} d_i\lambda_i \geq (1 - 1/\operatorname{e}) \OPT^\flex_\nAdap.\hspace{2cm}\qed\]
\end{proof}

\subsubsection{Tightness of the approximation factor.}
	The analysis of the algorithm given in \cref{thm:lopfn-approx} is tight. To see this, consider the following example instance of $\LOPfn$. Let $G$ be a directed cycle of length $n$, i.e., $V = \{v_1, \dots, v_n\}$ and $E = \{e_1, \dots, e_n\}$ with $e_i = (v_i, v_{i+1})$ for $i \in \{1, \dots, n-1\}$ and $e_n = (v_n, v_1)$. Let $K$ consist of $n$ commodities with unit demand, such that $s_i = v_i$ and $t_i = v_{i-1}$ for $i \in \{2, \dots, n\}$ and $s_1 = v_1$ and $t_1 = v_n$. Note $\mathcal{P}_i$ consists of a unique path of length $n-1$ for every $i \in K$. Finally, let $c = 0$, $F = T_i = 1$ for every $i \in K$ and $B = n/(n-1)$. Observe that the optimal solution of $\LPflex$ sets $p_e = 1/(n-1)$ for every $e \in E$. Using these probabilities as a solution to $\LOPfn$ yields a profit of $n \cdot \left(1 - (1 - \frac{1}{n-1})^{n-1}\right)$. On the other hand, setting $p_{e_1} = 1$ and $p_{e_n} = 1/(n-1)$ yields a profit of $n - 1 + 1/(n-1)$. Note that by choosing $n$ sufficiently large the ratio between these two values can be brought arbitrarily close to $1 - 1/\operatorname{e}$.
	
\subsubsection{Optimality gap when using $\LPflex$ for fixed fares.}
Unfortunately, $\LPflex$ does not yield an approximation guarantee for the fixed-fare setting. To see this, consider the following instance of $\LOP^\fixed_{X}$. There are four nodes $s_1, t_1, s_2, t_2$, together with edges from $s_1$ to $s_2$, $s_2$ to $t_2$, and $t_2$ to $t_1$, each with zero cost. In addition, there are $L$ parallel edges from $s_1$ to $t_1$, each with cost $1$. Commodity $1$ has origin $s_1$ and destination $t_1$ with a total demand of $L$. Commodity $2$ has origin $s_2$ and destination $t_2$ with a total demand of $1$. The budget is $1/2 + \varepsilon$ and the fine is $2$. Observe that the optimal solution of the corresponding instance of $\LPflex$ sets $p_{(s_2, t_2)} = 1/2 + \varepsilon$ and $p_e = 0$ for all other edges. Interpreting these probabilities as a solution of $\LOP\fixed_{X}$ results in a profit of $1 + 2\varepsilon$, as the followers represented by commodity $1$ will prefer one of the edges from $s_1$ to $t_1$ over the three-edge path. However, setting $p_{(s_2, t_2)} = 1/2$ instead yields a profit of $L + 1$.

\restateCorLpApprox*

\begin{proof}
	Again, let $(p, y)$ be an optimal solution to $\LPflex$.
	By \cref{lem:adaptivity_ratio_leader} and the proof of \cref{thm:lopfn-approx}, $\Gamma^{\flex}_{\adap,i}(p) \geq \frac{3}{4} \Gamma^{\flex}_{\nAdap, i}(p) \geq \frac{3}{4} (1 - 1/e) (y_i(t_i) - y_i(s_i) - \SP{c}{s_i, t_i})$ for all $i \in K$.\qed
\end{proof}

\section{Appendix for the Computational Study (\cref{section:heuristics})}\label{sec:appendix_c}

\subsection{Instances and Computational Setup} 

\subsubsection{Generation of Planar Graphs.} The graphs were generated using an approach similar to the one in~\cite{denise}. Vertices are distributed uniformly at random in the plane; iteratively, a vertex is chosen uniformly at random and connected to the closest neighbour that can be connected without violating planarity. This is repeated $3|V|-6$ times, after which disjoint connected components are connected using nearest Euclidean neighbours in order to ensure that the entire graph is connected. All arcs are present in both directions, as common in transit networks.
	
\subsubsection{Ticket Prices.}
For the randomly generated instances, ticket prices were set to
\[T_i = b+ m\cdot\frac{\SP{c}{s_i,t_i}}{\max_{v,w\in V}\SP{c}{v,w}},\] where $b$ is a base price and $b+m$ the maximum ticket price allowed in the network. This linear formula is a simplification of the formula used for official regulations regarding ticket prices in public transport networks~\cite{nloverheid,nlsite}. 

\subsubsection{Instance overview.}
Our study consist of seven graph sets, three based on real-world transit networks and four generated using the randomization procedure described above. Each set contains $40$ different combinations of graphs, customers and demands. For each of these combinations, twenty different budget values were tested, leading to $800$ different instances in each graph set; see \cref{table:instancesets} for average graph sizes.

\begin{table}[h]
	\begin{center}
		\begin{tabular}{r@{\hspace{0.2cm}}*{7}{@{\ }c@{\ }}}
		\hline
		 &  nlmajor & nlcomplete & adammetro & small & medium & large & huge \\
		\hline
		 $|V|$ & 23 & 341 & 45 & 25 & 50 & 100 & 200\\
		 $|E|$ & 60 & 864 & 88 & 95 & 195 & 399 & 806 \\
		\hline
		\end{tabular}
	\caption{Average data of test instances\label{table:instancesets}, showing the average  numbers of vertices and edges within each of the seven graph sets. }
	\end{center}\end{table}

\subsubsection{Implementation Details.}
All algorithms have been implemented in Java and compiled using jre7 on Windows 7 Enterprise. Computations have been performed on a machine with Intel Core 2 Duo CPU (GHz, 64 bit) and 4GB of memory using CPLEX 12.4 API for Java for the mathematical programs.

\subsection{Results}

\subsubsection{Computation times.}
\cref{tab:times} shows average computation times for the various algorithms. Using the LP for computing start solutions, a local optimum was reached within less than a minute for most instances, with the only exception being very large graphs combined with the (non-polynomial) non-adaptive followers' response. Furthermore, note that for large instances with many commodities, and for budgets higher than $6$, the mixed integer programming formulation could not be solved in reasonable time.

\begin{table}[h!]
	\begin{center}
		\begin{tabular}{r@{\hspace{0.4cm}}*{9}{c@{\hspace{0.2cm}}}}
		\hline
		graph set &LP&LP$^{\text{LS}}_N$&LP$^{\text{LS}}_A$&LPR&LPR$^{\text{LS}}_N$&LPR$^{\text{LS}}_A$&MIP&MIP$^{\text{LS}}_N$&MIP$^{\text{LS}}_A$  \\
		\hline
		nlmajor &0.17&10.2&1.59&1.39&10.6&1.68&6.31&9.93&1.62 \\
		nlcomplete &7.28&	931&21&	75.4&873&	19.7&	4024&563&	7.57 \\
		adammetro &0.33&	6.12&	1.39&	1.91	&9.81	&2.31&	3.29&	11.2&	2.87\\
		\hline
		small &0.04&0.75&0.11&0.29&0.8&0.11&0.6&0.77&0.11 \\
medium &0.06&4.5&0.18&0.32&4.9&0.19&1.5&3.90&0.18 \\
large &0.15&32.8&0.81&0.99&44.6&0.9&5.1&34.60&0.87 \\
huge &0.60&343.1&6.9&3.9&563.2&7.8&25.5&490.0&7.50 \\
\hline
		\end{tabular}
	\end{center}
	\caption{Average computation time in seconds\label{table:times} for fixed fares. 
	LP, LPR and MIP denote the solutions found by the corresponding mathematical program, and $\text{Alg}^{\text{LS}}_X$ with $\text{Alg} \in \{\text{LP},\text{LPR},\text{MIP}\}$ and $X \in \{\adap, \nAdap\}$ denotes the solution found by performing the local search heuristic on the solution found by algorithm $\text{Alg}$, using the followers response $X$.\vspace*{-1.3cm}}\label{tab:times}
\end{table}

\subsubsection{Influence of Budget.}
The leader's achievable revenue crucially depends on the available budget. \cref{figure:budgets} illustrates the quality of the solutions for the adammetro and nlcomplete instances, as well as for a small and large random instance, all with $50$ commodities for the $\LOPxn$ model, and the obtained revenue.

\begin{figure}
	\centering
	\begin{minipage}[t]{0.46\textwidth}
	\begin{tikzpicture}
	\begin{axis}[title={adammetro instance},height=4cm, width=7.2cm, no markers,
	legend pos=south east, xmax=6, xmin=0, ymin=0.8, xlabel=Budget,
every axis y label/.style={at={(current axis.north west)},above=2mm},
			ylabel=Gap]
	\addplot [color=red, thick][dash pattern=on 1pt off 2pt] table [x=Budget, y=Best, col sep=semicolon] {budgets_adammetro_Gaps.csv};
	\addplot [color=blue, thick][dash pattern=on 3pt off 4pt] table [x=Budget, y=LS+LPflexN, col sep=semicolon] {budgets_adammetro_Gaps.csv};
	\addplot [color=brown] table [x=Budget, y=LS+MCfixN, col sep=semicolon] {budgets_adammetro_Gaps.csv};
	\addplot [color=black, thick] table [x=Budget, y=Response MIPfixN, col sep=semicolon] {budgets_adammetro_Gaps.csv};
	\end{axis}
	\end{tikzpicture}	
\end{minipage}
\hspace{\fill}
\begin{minipage}[t]{0.46\textwidth}
		\begin{tikzpicture}
	\begin{axis}[title={adammetro instance},legend entries={Best,LP$^{LS}$, MC$^{LS}$, MIP},height=4cm, width=7.2cm, no markers, 
	legend pos = south east, legend style={font=\small},
	xmax=6, xmin=0, ymin=0, ylabel near ticks, yticklabel pos=right, xlabel=Budget,legend style={font=\scriptsize},
every axis y label/.style={at={(current axis.north east)},above=3mm},
			ylabel=Profit]
	\addplot [color=red, thick][dash pattern=on 1pt off 2pt]table [x=Budget, y=Best, col sep=semicolon] {budgets_adammetro.csv};
	\addplot [color=blue, thick][dash pattern=on 3pt off 4pt] table [x=Budget, y=LS+LPflexN, col sep=semicolon] {budgets_adammetro.csv};
	\addplot [color=brown]table [x=Budget, y=LS+MCfixN, col sep=semicolon] {budgets_adammetro.csv};
	\addplot [color=black, thick] table [x=Budget, y=Response MIPfixN, col sep=semicolon] {budgets_adammetro.csv};
	\end{axis}
	\end{tikzpicture}
\end{minipage}

\vspace{2mm}

\begin{minipage}[t]{0.46\textwidth}
	\begin{tikzpicture}
	\begin{axis}[title={nlcomplete instance},height=4cm, width=7.2cm, no markers,
	legend pos=south east, xmax=6, xmin=0, ymin=0.6,ymax=1, xlabel=Budget,
every axis y label/.style={at={(current axis.north west)},above=2mm},
			ylabel=Gap]
	\addplot [color=red, thick][dash pattern=on 1pt off 2pt] table [x=Budget, y=Best, col sep=semicolon] {budgets_nlcomplete_Gaps.csv};
	\addplot [color=blue, thick][dash pattern=on 3pt off 4pt] table [x=Budget, y=LS+LPflexN, col sep=semicolon] {budgets_nlcomplete_Gaps.csv};
	\addplot [color=brown] table [x=Budget, y=LS+MCfixN, col sep=semicolon] {budgets_nlcomplete_Gaps.csv};
	\addplot [color=black, thick] table [x=Budget, y=Response MIPfixN, col sep=semicolon] {budgets_nlcomplete_Gaps.csv};
	\end{axis}
	\end{tikzpicture}	
\end{minipage}
\hspace{\fill}
\begin{minipage}[t]{0.46\textwidth}
		\begin{tikzpicture}
	\begin{axis}[title={nlcomplete instance},height=4cm, width=7.2cm, no markers,
	legend pos = south east,
	xmax=6, xmin=0, ymin=0,ylabel near ticks, yticklabel pos=right,
every axis y label/.style={at={(current axis.north east)},above=3mm}, xlabel=Budget,
			ylabel=Profit]
	\addplot [color=red, thick][dash pattern=on 1pt off 2pt]table [x=Budget, y=Best, col sep=semicolon] {budgets_nlcomplete.csv};
	\addplot [color=blue, thick][dash pattern=on 3pt off 4pt] table [x=Budget, y=LS+LPflexN, col sep=semicolon] {budgets_nlcomplete.csv};
	\addplot [color=brown]table [x=Budget, y=LS+MCfixN, col sep=semicolon] {budgets_nlcomplete.csv};
	\addplot [color=black, thick] table [x=Budget, y=Response MIPfixN, col sep=semicolon] {budgets_nlcomplete.csv};
	\end{axis}
	\end{tikzpicture}
\end{minipage}
	
\vspace{2mm}

\begin{minipage}[t]{0.46\textwidth}
	\begin{tikzpicture}
	\begin{axis}[title={random instance $|V|=25$},height=4cm, width=7.2cm, no markers,
	legend pos=south east, xmax=6, xmin=0, ymin=0.8,ymax=1, xlabel=Budget,
every axis y label/.style={at={(current axis.north west)},above=2mm},
			ylabel=Gap]
	\addplot [color=red, thick][dash pattern=on 1pt off 2pt] table [x=Budget, y=Best, col sep=semicolon] {budgets_RndSmall_Gaps.csv};
	\addplot [color=blue, thick][dash pattern=on 3pt off 4pt] table [x=Budget, y=LS+LPflexN, col sep=semicolon] {budgets_RndSmall_Gaps.csv};
	\addplot [color=brown] table [x=Budget, y=LS+MCfixN, col sep=semicolon] {budgets_RndSmall_Gaps.csv};
	\addplot [color=black, thick] table [x=Budget, y=Response MIPfixN, col sep=semicolon] {budgets_RndSmall_Gaps.csv};
	\end{axis}
	\end{tikzpicture}	
\end{minipage}
\hspace{\fill}
\begin{minipage}[t]{0.46\textwidth}
		\begin{tikzpicture}
	\begin{axis}[title={random instance $|V|=25$},height=4cm, width=7.2cm, no markers,
	legend pos = south east,
	xmax=6, xmin=0, ymin=0,ylabel near ticks, yticklabel pos=right, xlabel=Budget,
every axis y label/.style={at={(current axis.north east)},above=3mm},
		ylabel=Profit]
	\addplot [color=red, thick][dash pattern=on 1pt off 2pt]table [x=Budget, y=Best, col sep=semicolon] {budgets_RndSmall.csv};
	\addplot [color=blue, thick][dash pattern=on 3pt off 4pt] table [x=Budget, y=LS+LPflexN, col sep=semicolon] {budgets_RndSmall.csv};
	\addplot [color=brown]table [x=Budget, y=LS+MCfixN, col sep=semicolon] {budgets_RndSmall.csv};
	\addplot [color=black, thick] table [x=Budget, y=Response MIPfixN, col sep=semicolon] {budgets_RndSmall.csv};
	\end{axis}
	\end{tikzpicture}
\end{minipage}
		
\vspace{2mm}

\begin{minipage}[t]{0.46\textwidth}
	\begin{tikzpicture}
	\begin{axis}[title={random instance $|V|=100$},height=4cm, width=7.2cm, no markers,
	legend pos=south east, xmax=6, xmin=0, ymin=0.55, ymax=1, xlabel=Budget,
every axis y label/.style={at={(current axis.north west)},above=2mm},
			ylabel=Gap]
	\addplot [color=red, thick][dash pattern=on 1pt off 2pt] table [x=Budget, y=Best, col sep=semicolon] {budgets_RndLarge_Gaps.csv};
	\addplot [color=blue, thick][dash pattern=on 3pt off 4pt] table [x=Budget, y=LS+LPflexN, col sep=semicolon] {budgets_RndLarge_Gaps.csv};
	\addplot [color=brown] table [x=Budget, y=LS+MCfixN, col sep=semicolon] {budgets_RndLarge_Gaps.csv};
	\addplot [color=black, thick] table [x=Budget, y=Response MIPfixN, col sep=semicolon] {budgets_RndLarge_Gaps.csv};
	\end{axis}
	\end{tikzpicture}	
\end{minipage}
\hspace{\fill}
\begin{minipage}[t]{0.46\textwidth}
		\begin{tikzpicture}
	\begin{axis}[title={random instance $|V|=100$},height=4cm, width=7.2cm, no markers,
	legend pos = south east,
	xmax=6, xmin=0, ymin=0,ylabel near ticks, yticklabel pos=right, xlabel=Budget,
every axis y label/.style={at={(current axis.north east)},above=3mm},
		ylabel=Profit]
	\addplot [color=red, thick][dash pattern=on 1pt off 2pt]table [x=Budget, y=Best, col sep=semicolon] {budgets_RndLarge.csv};
	\addplot [color=blue, thick][dash pattern=on 3pt off 4pt] table [x=Budget, y=LS+LPflexN, col sep=semicolon] {budgets_RndLarge.csv};
	\addplot [color=brown]table [x=Budget, y=LS+MCfixN, col sep=semicolon] {budgets_RndLarge.csv};
	\addplot [color=black, thick] table [x=Budget, y=Response MIPfixN, col sep=semicolon] {budgets_RndLarge.csv};
	\end{axis}
	\end{tikzpicture}
\end{minipage}	
	\caption{Graphs illustrating the influence of the budget on the solution quality, for the adammetro instance, the nlcomplete instance, a random instance with $25$ vertices and a random instance with $100$ vertices, each with $50$ commodities for the $\LOPxn$ model. Note that the legend in the top-right graph applies to all graphs. Graphs on the left show gaps relative to the upper bound, graphs on the right show the corresponding profit. 'Best' denotes the best solution found, $\text{LP}^{\text{LS}}$ and $\text{MC}^{\text{LS}}$ denote the solutions found by performing the local search heuristic on the LP solution resp. MultiCut solution, and MIP denotes the solution found by the corresponding mathematical program from~\cite{borndorfer}.}\label{figure:budgets}

\end{figure}
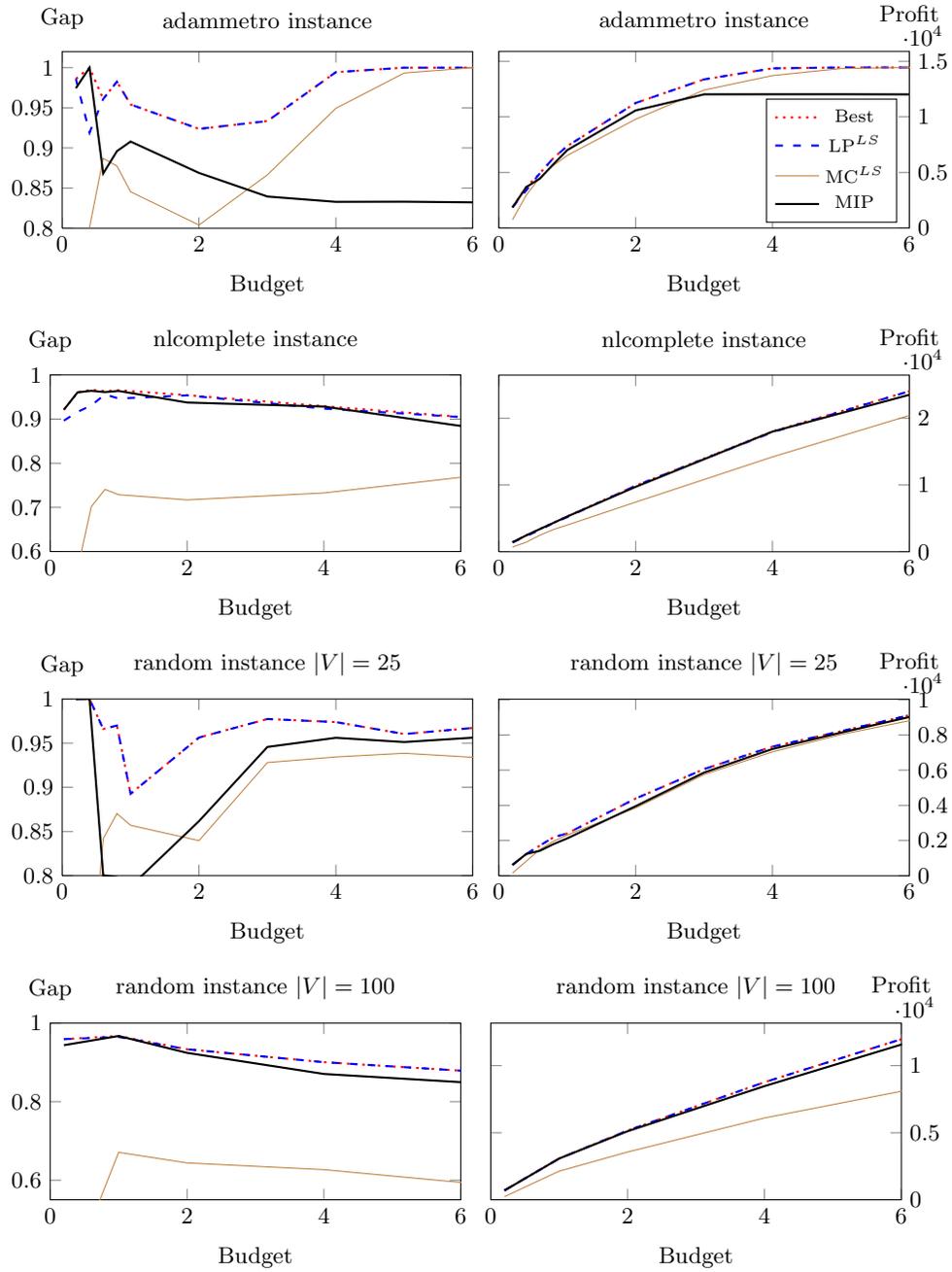

\subsubsection{Complete overview of solution quality.}
\cref{table:results1,table:results2,table:results3} show average ratios between solutions and upper bounds for all models. Here, \textbf{LPR} denotes the linear relaxation of the mixed integer program from~\cite{borndorfer} (\textbf{MIP}).

\begin{table}[t]
	\begin{center}
		\begin{tabular}{r@{\hspace{0.5cm}}*{4}{c@{\hspace{0.2cm}}}@{\hspace{0.2cm}}*{4}{c@{\hspace{0.2cm}}}}
		\hline
		&\multicolumn{4}{c}{$\LOPfn$}&\multicolumn{4}{c}{$\LOPfa$} \\
		graph set &LP&LP$^{\text{LS}}$&MC&MC$^{\text{LS}}$ &LP&LP$^{\text{LS}}$&MC&MC$^{\text{LS}}$  \\
		\hline
		nlmajor &96.4	&98.0&	52.8	&85.2	&96.4	&98.0&	52.8&	85.2 \\
		nlcomplete &97.8	&97.9	&	33.9	&39.7	 &	97.7&	97.9&	33.9	&39.7\\
		adammetro &94.7&	97.8	&69.0&	86.7 &94.7	&97.8	&69.0&	86.7\\
		\hline
		small &95.3&97.5&49.1&77.7 &95.3&97.4&49.0&77.5 \\
medium &95.5&97.5 &41.6&69.4 &95.4&97.3&41.6&69.3 \\
large &96.6&98.0 &40.7&75.9& 96.5&97.8&40.7&75.8 \\
huge &96.5&97.9 &37.2&71.1 &96.4&97.7&37.2&71.2 \\
  \hline
		\end{tabular}
	\end{center}
	\caption{Average ratios between solutions and upper bounds in percent for test instances in the flexible-fare setting. $\text{LP}^{\text{LS}}$ and $\text{MC}^{\text{LS}}$ denote the solutions found by performing the local search heuristic on the LP solution and the MultiCut solution, respectively.\label{table:results1}}
\end{table}

\begin{table}[t]
	\begin{center}
		\begin{tabular}{r@{\hspace{0.4cm}}*{8}{c}}
		\hline
		graph set &LP&LP$^{\text{LS}}$&LPR&LPR$^{\text{LS}}$&MIP&MIP$^{\text{LS}}$&MC&MC$^{\text{LS}}$  \\
		\hline
		nlmajor &94.1	&97.2&	94.8&	97.2	&94.7&	97.3&	52.2	&85.1\\
		nlcomplete &90.9	&92.6	&90.6&	92.4&	91.0&	93.3&	33.3	&37.2\\
		adammetro &94.7	&97.8	&92.0	&97.5	&91.7&	97.3	&69.0&86.7\\
		\hline
		small &88.1&96.0&89.9&96.3&89.8&96.3&48.7&74.1 \\
medium &88.9&95.0&90.1&95.0&90.5&95.1&41.1&64.9 \\
large &88.6&94.5&89.0&94.9&89.0&95.1&40.3&71.6 \\
huge &85.2&92.7&87.8&93.3&88.1&94.0&36.8&64.8 \\
\hline
		\end{tabular}
	\end{center}
	\caption{Average ratios between solutions and upper bounds in percent for test instances with fixed ticket prices and adaptive followers. LP, LPR and MIP denote the solutions found by the corresponding mathematical program, and $\text{Alg}^{\text{LS}}$ with $\text{Alg} \in \{\text{LP},\text{LPR},\text{MIP}\}$ denotes the solution found by performing the local search heuristic on the solution found by algorithm $\text{Alg}$.\label{table:results2}}
\end{table}

\newpage

\begin{table}[t!]
	\begin{center}
		\begin{tabular}{r@{\hspace{0.4cm}}*{8}{c}}
		\hline
		graph set &LP&LP$^{\text{LS}}$ &LPR&LPR$^{\text{LS}}$&MIP&MIP$^{\text{LS}}$ &MC&MC$^{\text{LS}}$  \\
		\hline
		nlmajor &94.8&	97.3	&95.1	&97.3	&95.4	&97.4&	52.2	&85.1\\
		nlcomplete & 91.6&	93.2 &	91.9&	93.6	&93.4	&94.4 &33.3	&37.1\\
		adammetro & 94.7	&97.8	 &92.0&	97.5 &91.7&	97.3 &69.0&	86.7\\
		\hline
		small &92.1&96.3 &92.7&96.5 &92.8&96.4 &48.7&74.2 \\
medium &90.7&95.2 &91.8&95.4 &91.9&95.6& 41.1&64.9\\
large &90.4&95.1 &90.8&95.2 &91.3&95.4& 40.3&71.7\\
huge &89.1&94.1& 91.5&94.5 &91.9&95.3& 36.8&64.8\\
\hline
		\end{tabular}
	\end{center}
	\caption{Average ratios between solutions and upper bounds in percent for test instances with fixed ticket prices and non-adaptive followers. LP, LPR and MIP denote the solutions found by the corresponding mathematical program, and $\text{Alg}^{\text{LS}}$ with $\text{Alg} \in \{\text{LP},\text{LPR},\text{MIP}\}$ denotes the solution found by performing the local search heuristic on the solution found by algorithm $\text{Alg}$.\label{table:results3}}
\end{table}	
                             
\end{document}